\documentclass{article}

\usepackage{PRIMEarxiv}

\usepackage[utf8]{inputenc} 
\usepackage[T1]{fontenc}    
\usepackage{hyperref}       
\usepackage{url}            
\usepackage{booktabs}       
\usepackage{amsfonts}       
\usepackage{nicefrac}       
\usepackage{microtype}      
\usepackage{lipsum}
\usepackage{fancyhdr}       
\usepackage{graphicx}       
\graphicspath{{media/}}     

\usepackage{multirow}%
\usepackage{amsmath,amssymb}%
\usepackage{amsthm}%
\usepackage{mathrsfs}%
\usepackage[title]{appendix}%
\usepackage{xcolor}

\newcommand{\rev}[1]{#1}
\newcommand{\changed}[1]{#1}
\usepackage{textcomp}%
\usepackage{manyfoot}%
\usepackage{algorithm}%
\usepackage{algorithmicx}%
\usepackage{algpseudocode}%
\usepackage{listings}%
\usepackage{mathtools}
\usepackage{tikz}
\usetikzlibrary{shadows, arrows.meta, positioning, decorations.pathreplacing}
\usepackage{subcaption}
\usepackage{array}
\usepackage{bbm, dsfont}
\usetikzlibrary{calc,backgrounds}
 \pgfdeclarelayer{bg}
 \pgfsetlayers{bg,main}

\newcommand{\opt}{\text{OPT}}
\newcommand{\sw}{\text{SW}}

\newtheorem{theorem}{Theorem}
\newtheorem{lemma}{Lemma}
\newtheorem{proposition}{Proposition}
\newtheorem{corollary}{Corollary}
\newtheorem{definition}{Definition}
\newtheorem{remark}{Remark}
\newtheorem{example}{Example}

\usepackage{siunitx}
\title{\large Latency-Aware Resource Allocation over Heterogeneous Networks: \\
A Lorentz-Invariant Market Mechanism
}

\author{
  Saad Alqithami \\
  \texttt{alqithami@gmail.com} \\
}

\begin{document}
\maketitle

\begin{abstract}
We present a telecom-native auction mechanism for bandwidth and time-slot allocation across heterogeneous-delay networks, from low-Earth-orbit (LEO) satellite constellations to delay-tolerant deep-space relays. The Lorentz-Invariant Auction (LIA) treats bids as spacetime events and reweights reported values by the \emph{horizon slack}---a causal quantity derived from earliest-arrival times relative to a public clearing horizon. The contribution is not merely another delay-equalization rule: relative to batching, synchronized clearing, and speed-bump designs, LIA combines a causal-ordering formulation, a uniquely exponential slack correction implied by a semigroup-style invariance axiom, and a critical-value implementation that is truthful in reported values once slacks are fixed by trusted infrastructure. We therefore state the incentive result in the exogenous-slack regime and separately analyze bounded slack-estimation error and endogenous-delay limitations. Under fixed feasible slacks, LIA is individually rational and achieves welfare at least \(e^{-\lambda\Delta}\) relative to the optimal feasible allocation, where \(\Delta\) is the slack spread. We evaluate LIA on STARLINK-200, INTERNET-100, and DSN-30 across 52{,}500 baseline instances with market sizes \(n\in\{10,20,30,40,50\}\) and we run additional robustness sweeps. On Starlink and Internet, LIA remains near-efficient while eliminating measured timing rents; on DSN, welfare is lower in thin markets but improves with depth. We also distinguish winner-determination time from the background cost of maintaining slack estimates and study robustness beyond iid noise through error-spread bounds and structured (distance-biased and subnetwork-correlated) noise models. These results suggest that causal-consistent mechanism design offers a practical non-buffering alternative to synchronized delay equalization in heterogeneous telecom infrastructures.
\end{abstract}

\keywords{telecommunication systems\and auction-based resource allocation\and network latency\and truthful mechanisms\and fairness\and causal consistency.}

\section{Introduction}\label{sec1}

\subsection{Motivation and Context}

\changed{Telecommunication markets increasingly operate over infrastructures where one-way delays vary by orders of magnitude. A bid transmitted along a Chicago–New Jersey microwave route arrives within microseconds, whereas the same message forwarded through a low-Earth-orbit (LEO) constellation such as Starlink experiences tens of milliseconds of propagation delay \cite{handley2018delay,mohan2024starlink}. Recent empirical measurements of Starlink one-way delays \cite{garcia2025starlink} and statistical characterizations of end-to-end latency \cite{casparsen2026statistical} confirm that these delays are not merely random noise but exhibit structured, predictable periodicity. At planetary and interplanetary scales, the relevant delays extend to minutes or hours \cite{burleigh2003delay,yan2024routing}. Beyond technical performance, LEO satellite Internet is also intertwined with policy and geopolitical considerations \cite{remili2025tech}, and plays a foundational role in the roadmap toward 6G and beyond \cite{ntontin2025vision}. Despite this heterogeneity, many market-clearing rules still presume near-instantaneous communication and a unitary global timeline for collecting and ranking bids. Recent surveys in market design emphasize that deployment contexts in which communication constraints influence incentives and outcomes are no longer pathological edge cases; they are now central to the theory–practice interface for mechanisms that must live on real networks \cite{felicetti2025systematic,luong2026incentive}.}

\changed{The heterogeneity of communication delays in modern telecommunication networks presents significant challenges for resource allocation mechanisms. In a Starlink-like constellation with 1,500+ satellites orbiting at 550 km altitude, inter-satellite links experience delays ranging from 1.8 ms (adjacent satellites in the same orbital plane) to 47.3 ms (satellites in different planes on opposite sides of Earth). When allocating bandwidth for data transmission, satellites closer to the allocation decision point have a systematic advantage in traditional auction mechanisms, as their bids arrive earlier and can be processed before those from more distant satellites. Similarly, in a global internet backbone spanning multiple continents, one-way propagation delays between major nodes can range from 0.3 ms (co-located data centers) to 89.2 ms (transcontinental links). When allocating computational resources at edge servers, clients located near the servers have a significant timing advantage over distant clients, potentially leading to inefficient resource allocation and unfair advantages based solely on geographic proximity. At the extreme end, NASA's Deep Space Network, extended with hypothetical interplanetary relay stations, experiences one-way communication delays ranging from 8.3 minutes (Earth-Mars at closest approach) to 73.6 minutes (Earth-Jupiter at maximum separation). When scheduling communication windows for spacecraft, the extreme delay heterogeneity makes traditional synchronous allocation mechanisms practically unusable, as waiting for all bids to arrive would introduce unacceptable latency.}

\changed{These examples illustrate the fundamental tension between the assumptions of classical mechanism design and the physical realities of distributed communication systems. Classical auction theory, exemplified by the seminal work of Vickrey, Clarke, and Groves \cite{clarke1971multipart}, assumes that all participants can submit bids simultaneously and that the auctioneer has instantaneous access to all submitted bids. This assumption becomes increasingly problematic as communication networks span greater distances and exhibit more heterogeneous latency characteristics \cite{qiu2024settling,yang2018designing}. A perspective grounded in network physics and distributed systems makes clear why the global-clock abstraction is brittle. In special relativity, each participant observes events through a local causal frame, and simultaneity is frame dependent \cite{einstein1905electrodynamics,minkowski1952space}. Two bid emissions that appear simultaneous to one observer can have a definite temporal ordering for another. Distributed systems arrived at an analogous insight decades ago: causal consistency, rather than global synchronization, is the right primitive for reasoning about computations and decisions over links with nontrivial propagation \cite{lamport1978time,fidge1988timestamps,mattern1988time}. When decision and reaction times are comparable to signal velocities, the assumption of a single total order enforced by a global clock becomes both impractical as an engineering matter and inconsistent with the causal structure that constrains what information can influence a clearing decision. Modern fault-tolerant stacks further explore these trade-offs, but their goal is total-order replication rather than welfare-maximizing allocation \cite{zhong2023byzantine,tang2024improved}.}

\subsection{Challenges in Heterogeneous-Delay Environments}

\changed{The classical Vickrey–Clarke–Groves family of truthful mechanisms presupposes instantaneous awareness of competitors' bids and leans on a total temporal order to resolve conflicts \cite{clarke1971multipart}. Related optimal-auction results inherit the same synchrony idealization \cite{myerson1981optimal}. In regimes where propagation delays are material, these assumptions fail in structured ways that generate welfare losses, distributional concerns, and profitable timing games.}

\changed{The electronic-markets literature documents how small timing advantages are monetized and how they distort allocation outcomes \cite{aquilina2022quantifying,nber2025ai}. Engineering fixes such as speed bumps and frequent batch auctions pursue fairness by injecting controlled waiting, yet they retain a synchronous worldview and therefore treat symptoms rather than causes \cite{khapko2021speed,budish2015high}. Recent studies comparing frequent batch auctions to continuous-time auctions confirm that while batching benefits liquidity, it can harm market volatility \cite{ge2025frequent}. If the fundamental friction is that information cannot outrun light, then remedies that force a semblance of simultaneity by holding bids in buffers are at best partial and at worst counterproductive when delays are not uniform. This physical constraint has recently inspired formal treatments of relativistic asset pricing, which define causal domains analogous to light cones to model financial markets constrained by the finite speed of information \cite{cheah2025relativistic}.}

\changed{This observation helps position the novelty of LIA more precisely. Frequent batch auctions, common-deadline normalization, speed-bump delay equalization, and asynchronous market designs all try to reduce temporal advantage, but they do so by changing when the market clears or by delaying messages. LIA instead takes \emph{causal reachability} as primitive and changes how bids are compared at a public horizon, aligning with emerging relativistic frameworks for market microstructure \cite{cheah2025relativistic}. The distinctive claim is therefore the combination of (i) a causal-ordering formulation based on earliest arrival, (ii) a uniquely exponential slack correction, and (iii) a truthful critical-value implementation in the single-parameter regime induced by fixed exogenous slacks, which complements recent advances in truthful mechanism design for distributed networked agents \cite{zhong2026truthful}. (We provide a concrete numerical example of how this correction neutralizes timing rents in Example~\ref{ex:timing_rent}.)}

\changed{The challenges of resource allocation in heterogeneous-delay environments can be categorized into three main areas. First, fairness concerns arise because participants located near clearing infrastructure receive systematic priority solely because light travels to them faster. This geographic advantage is unrelated to the intrinsic value they place on the resource, leading to allocative inefficiencies and unfair outcomes. In financial markets, this phenomenon has led to the well-documented "latency arbitrage" problem, where market participants invest heavily in reducing communication delays to gain trading advantages \cite{aquilina2022quantifying,khapko2021speed}. Second, efficiency losses occur when allocation decisions do not account for heterogeneous delays, as resources may be assigned to participants who are not the highest-value users but simply the fastest to communicate. This leads to welfare losses that can be significant, especially in networks with high delay dispersion. Traditional mechanisms that attempt to address this by waiting for all bids to arrive introduce excessive latency and reduce throughput. Third, implementation challenges exist because practical deployment of allocation mechanisms in telecommunication networks requires algorithms that can operate efficiently at scale, often processing thousands of bids per second with minimal computational overhead. The mechanism must also be robust to measurement errors, network jitter, and potential strategic manipulation of reported delays.}

\subsection{Single-Item Model and Telecom Resources}

\changed{Each clearing instance in our model corresponds to a single indivisible time-slot, frequency band, or computational resource block that cannot be subdivided among multiple winners. This abstraction captures many practical scenarios in telecommunications: allocation of exclusive access to a satellite transponder during a specific time window, assignment of a dedicated network slice with guaranteed quality-of-service parameters, or reservation of computational resources at an edge server for latency-critical applications \cite{dalai2024novel,zhang2022scalable}. While many telecom resource allocation problems involve multiple items or combinatorial preferences, the single-item case represents a fundamental building block that must be solved efficiently before addressing more complex scenarios. We discuss extensions to multiple identical items and combinatorial settings in Section~\ref{sec:extensions}.}

\changed{The single-item model is particularly relevant in contexts such as satellite communication windows, where multiple ground stations compete for exclusive access to a satellite's communication capabilities during a specific time window, and only one station can be granted access at a time. It also applies to spectrum allocation in dynamic spectrum access systems, where a specific frequency band may be allocated to a single user for a limited time period, requiring efficient and fair allocation mechanisms. Furthermore, in computational offloading, edge computing resources may be allocated to process a single computationally intensive task at a time, with multiple users competing for these limited resources. Finally, in network function virtualization, virtual network functions may require exclusive access to specific hardware resources for a time window, necessitating allocation mechanisms that account for heterogeneous communication delays.}

\changed{This paper adopts causal consistency as a design primitive for resource allocation in telecommunication systems. We treat bandwidth and time-slot markets as distributed economic processes that run on top of links with finite signal velocity, including LEO inter-satellite networks, terrestrial backbones, and delay-tolerant deep-space relays. The organizing question is how to allocate scarce communication resources in a way that is incentive compatible and welfare efficient while remaining fair across heterogeneous delays and avoiding artificial waiting. The answer we develop is to replace global time with proper time measured along world-lines, to evaluate bids at a clearing horizon using only causally reachable information, and to translate causal reachability into an algorithmic structure that can be executed at scale without trusted synchronization services.}

\changed{We introduce the Lorentz-Invariant Auction (LIA), a mechanism that embeds bid events in Minkowski spacetime and discounts nominal values by the proper-time lapse from emission to the clearing horizon. The correction neutralizes proximity and propagation advantages by penalizing bids that have enjoyed more causal slack before the decision point. Because both proper-time lapses and monetary amounts are scalars, the ranking rule is frame independent \cite{einstein1905electrodynamics,minkowski1952space}. The precise mechanism-design claim is narrower than the original draft suggested: LIA is dominant-strategy truthful in \emph{reported values} once the slack inputs are fixed by trusted infrastructure, and its welfare guarantee is naturally stated relative to the best feasible bid. Figure~\ref{figspacetime} gives the geometric intuition for this correction: bids with larger horizon slack have had more opportunity to reach or react before the public horizon, so LIA discounts them more heavily at the allocation frontier. We also establish that any discount function satisfying the stated invariance and composition properties must be exponential, which is why the mechanism uses \(e^{-\lambda\delta_i}\) rather than an ad hoc linear or polynomial correction.}

\changed{Beyond conceptual alignment with causal structure, the approach admits an efficient implementation tailored to real-time clearing in heterogeneous-delay networks. For a fixed topology (or a snapshot updated on a control-plane timescale), we precompute shortest-path propagation delays to the clearing horizon and compute each bid's horizon slack in $O(1)$. For the single-item auction studied here, allocation and payments reduce to a single pass over discounted bids---$O(n)$ time and $O(1)$ additional space---making the mechanism practical for high-frequency deployment.}

\begin{figure}[t]
  \centering

\begin{tikzpicture}[x=1cm,y=1cm,>=Latex,line cap=round]
  \draw[->] (-0.8,0) -- (5.5,0) node[below right]{space};
  \draw[->] (0,-0.6) -- (0,4.8) node[above]{time};

  \draw[very thick] (-0.3,4) -- (5.3,4) node[right] {$H$};

  \path[fill=gray!20] (2,2.05) rectangle (2.7,2.85);
  \node at (2.35,2.65) {\footnotesize Slab};
  
  \draw[thick,->] (1,0) -- (1,4) node[above,yshift=6pt]{\footnotesize Bidder 1};
  \draw[thick,->] (3.6,0) -- (3.3,4) node[above,yshift=6pt]{\footnotesize Bidder 2};

  \fill (1,1) circle (1.6pt) node[below left=-2pt] {$e_1$};
  \fill (3.5,1.2) circle (1.6pt) node[below right=-2pt] {$e_2$};

  \draw[dashed] (1,1) -- ++( 2,2);
  \draw[dashed] (1,1) -- ++(-2,2);
  \draw[dashed] (3.5,1.2) -- ++( 2,2);
  \draw[dashed] (3.5,1.2) -- ++(-2,2);

  \draw[->] (1,1) -- (1,4) node[midway,right=2pt] {$\delta_1$};
  \draw[->] (3.5,1.2) -- (3.3,4) node[midway,left=2pt] {$\delta_2$};

  \draw[very thick,red]
    (0.2,2.35) .. controls (2.5,2.2) and (3.6,2.05) .. (5.0,1.9)
      node[pos=0.78,below,sloped,fill=white,inner sep=1.2pt,text=red]
      {\footnotesize allocation frontier};
\end{tikzpicture}

  \caption{Spacetime view of two bids. Dashed lines are light-cone boundaries under unit aspect ($c{=}1$). The shaded rectangle marks a causal slab where bids are incomparable; the red line is the frontier where horizon-discounted bids are compared (i.e., where $\theta_1 e^{-\lambda \delta_1}=\theta_2 e^{-\lambda \delta_2}$). Its position depends on the values $(\theta_1,\theta_2)$, the parameter $\lambda$, and how each bid's horizon slack $\delta_i$ changes with emission time. Proper-time lapses $\delta_i$ (horizon slacks) are measured along each worldline from emission to the clearing horizon $H$, with $\delta_i \equiv \Delta\tau_i$.}
  \label{figspacetime}
\end{figure}
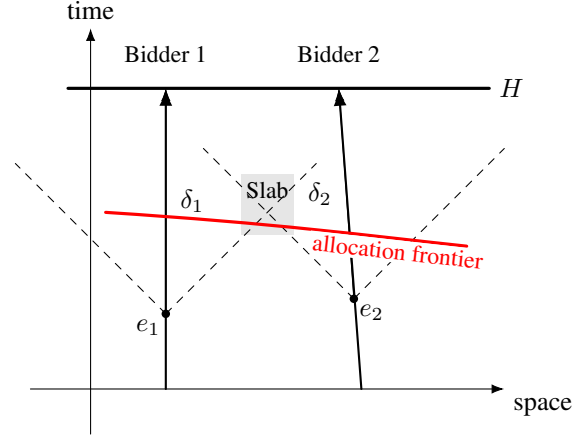

\subsection{Our Contributions}

\changed{We make four primary contributions to the study of resource allocation in heterogeneous-delay networks. First, we introduce a spacetime formalism for mechanism design, modeling bids as \emph{spacetime events} and defining the \emph{horizon slack}---the causal margin by which a bid can still reach a public clearing horizon under earliest-arrival routing. We prove that, under mild invariance and time-consistency axioms, the only discount function compatible with this causal slack is exponential. Second, we propose the Lorentz-Invariant Auction (LIA), which discounts values by horizon slack and then runs a truthful VCG-style allocation and payment rule on the discounted bids. The formal incentive theorem is stated in the correct single-parameter regime: truthful value reporting is dominant and the mechanism is individually rational \emph{conditional on fixed exogenous slacks}. We also establish a welfare guarantee relative to the best feasible bid, with a loss factor of \(1-e^{-\lambda\Delta}\) governed by slack dispersion, and we formalize a \emph{Latency-Arbitrage Index} (LAI) that measures the marginal utility gain from reducing propagation delay. Third, we provide an efficient implementation of LIA using standard shortest-path delay estimation to compute slacks. For the single-item auction studied here, winner determination and payments require only a single pass over discounted bids ($O(n)$ given precomputed delays). In our prototype, winner determination runs in $\mathcal{O}(10)\,\mu\mathrm{s}$ for $n\le 50$ and remains under $0.2$\,ms at $n=1000$. Finally, we conduct a comprehensive empirical evaluation and robustness analysis of LIA on STARLINK-200, INTERNET-100, and DSN-30 across hundreds of thousands of Monte Carlo instances. The experiments explicitly situate LIA against waiting-based fairness designs---Sync-VCG/common-deadline normalization, HoldBack/speed-bump equalization, and Batch-VCG/frequent batch clearing---and we distinguish auction runtime from the background cost of maintaining the slack map. On STARLINK-200 and INTERNET-100 with \(n\in\{10,20,30,40,50\}\), LIA attains welfare ratios of approximately \(0.996\)--\(0.997\) while driving the estimated timing rent (LAI) to zero; on DSN-30 the small-market regime is harsher, but scaling to \(n=1000\) yields a welfare ratio of about \(0.996\). Robustness sweeps with up to \(10\)\,ms slack-measurement error (iid, systematic clock-bias, distance-biased, and subnetwork-correlated models) show negligible welfare degradation and no re-emergence of timing rents.}

\subsection{Related Work}

\subsubsection{Auction Theory and Mechanism Design}

\changed{The theoretical foundation of our work builds upon classical auction theory and mechanism design. The Vickrey-Clarke-Groves (VCG) mechanism \cite{clarke1971multipart} provides a framework for designing truthful mechanisms in quasi-linear environments. Myerson's optimal auction design \cite{myerson1981optimal} characterizes revenue-maximizing mechanisms for single-parameter settings. These foundational works, however, assume synchronous communication and instantaneous bid processing, which become problematic in networks with heterogeneous delays.}

\changed{Recent work has begun to address communication constraints in mechanism design. Qiu and Weinberg \cite{qiu2024settling} study the communication complexity of VCG-based mechanisms, establishing lower bounds on the amount of information that must be exchanged to achieve various approximation guarantees. Yang et al. \cite{yang2018designing} propose distributed auction mechanisms for wireless spectrum allocation, but their approach still relies on synchronization primitives that may not scale to highly heterogeneous delay environments.}

\changed{Felicetti et al. \cite{felicetti2025systematic} provide a systematic review of auction mechanisms, highlighting the growing importance of deployment contexts where communication constraints influence incentives and outcomes. Their work emphasizes the need for mechanisms that can operate efficiently in realistic network environments, a challenge that our Lorentz-Invariant Auction directly addresses.}

\subsubsection{Latency Arbitrage and Market Microstructure}

\changed{The financial markets literature has extensively studied the phenomenon of latency arbitrage, where market participants gain advantages through faster access to information. Aquilina et al. \cite{aquilina2022quantifying} quantify the "arms race" in high-frequency trading, documenting substantial investments in reducing communication latency to gain trading advantages. Dou et al. \cite{nber2025ai} examine how AI-powered trading affects algorithmic collusion and price efficiency, highlighting the role of latency in strategic behavior.}

\changed{Various remedies have been proposed to address latency arbitrage in financial markets. Khapko and Zoican \cite{khapko2021speed} study the effectiveness of "speed bumps" that introduce artificial delays to equalize arrival times. Budish et al. \cite{budish2015high} propose frequent batch auctions as an alternative market design that processes trades in discrete time intervals rather than continuously. While these approaches mitigate some aspects of latency arbitrage, they retain a synchronous worldview and may introduce inefficiencies when delays are highly heterogeneous.}

\changed{Our work differs from these approaches by directly incorporating causal structure into the mechanism design, rather than attempting to enforce synchrony through artificial delays. This allows LIA to address latency fairness while maintaining efficiency and incentive compatibility.}

\changed{To clarify the mechanism design space, we can systematically position LIA against existing approaches to latency in electronic markets:
\begin{itemize}
    \item \textbf{Continuous Clearing (e.g., Fast-VCG):} Clears immediately upon bid arrival. Minimizes latency but maximizes timing rents, heavily favoring proximate bidders.
    \item \textbf{Frequent Batch Auctions (Batch-VCG):} Buffers bids into fixed time windows. Reduces timing rents but imposes artificial waiting on all participants, tracing a Pareto frontier between fairness and latency.
    \item \textbf{Speed Bumps / Delay Equalization (HoldBack):} Artificially delays fast connections to match the slowest participant. Achieves fairness but forces the entire system to operate at the speed of the most distant node.
    \item \textbf{Lorentz-Invariant Auction (LIA):} Uses causal-consistent pricing rather than buffering. Achieves perfect fairness (zero timing rent) without artificial holdbacks, by mathematically discounting the value of the timing advantage itself.
\end{itemize}
This structural difference is why LIA can maintain high welfare and strict fairness even under structured network noise, as demonstrated in our empirical evaluation.}

\subsubsection{Distributed Systems and Causal Consistency}

\changed{The theoretical foundation of our work draws heavily from distributed systems research on causal consistency and partial ordering of events. Lamport's seminal paper \cite{lamport1978time} introduced the happens-before relation, which provides a partial ordering of events in a distributed system without requiring global time synchronization. This concept was further developed through vector clocks \cite{fidge1988timestamps,mattern1988time}, which provide a practical mechanism for tracking causal dependencies.}

\changed{Schwarz and Mattern \cite{schwarz1994detecting} explore the challenge of detecting causal relationships in distributed computations, highlighting the fundamental difficulty of establishing a total order on events in a distributed system. Recent work by Zhao et al. \cite{zhao2024detecting} applies machine learning techniques to detect dynamical causality in complex systems, demonstrating the continued relevance of causal analysis in distributed environments.}

\changed{The connection between causality and fault tolerance has been extensively studied in the distributed systems literature. Zhong et al. \cite{zhong2023byzantine} survey Byzantine fault-tolerant consensus algorithms, which must operate correctly despite arbitrary failures and communication delays. Tang et al. \cite{tang2024improved} propose an improved Byzantine fault-tolerant consensus mechanism that leverages causal relationships to enhance efficiency and robustness.}

\changed{Our contribution extends this causal perspective to mechanism design, where the challenge is not merely to maintain consistency across distributed computations, but to ensure that economic incentives remain aligned in the presence of heterogeneous communication delays.}

\subsubsection{Telecommunication Networks and Resource Allocation}

\changed{The practical motivation for our work comes from the challenges of resource allocation in modern telecommunication networks. Remili et al. \cite{remili2025tech} examine the geopolitical implications of LEO satellite internet constellations, highlighting the growing importance of these networks in global communications infrastructure. Zhang and Yeung \cite{zhang2022scalable} propose scalable routing algorithms for LEO satellite constellations, addressing the challenges of rapidly changing network topologies and heterogeneous link characteristics.}

\changed{Dalai et al. \cite{dalai2024novel} present a novel space-based hosting approach for ultra-low-latency web services, demonstrating the potential of satellite networks to support latency-sensitive applications. Yan et al. \cite{yan2024routing} develop routing protocols for future space-terrestrial integrated networks, addressing the challenges of integrating satellite and terrestrial communication systems. Furthermore, recent surveys highlight the growing importance of incentive mechanisms—including auction-based, game-based, and learning-based approaches—for resource management in satellite networks \cite{luong2026incentive}. For instance, auction-based frameworks have been proposed for dynamic pricing and scheduling in LEO networks \cite{li2025dynamic}.}

\changed{At the extreme end of the delay spectrum, Burleigh et al. \cite{burleigh2003delay} introduce delay-tolerant networking as an approach to interplanetary Internet, addressing the challenges of communication in environments where delays can extend to minutes or hours. These works highlight the diverse range of delay characteristics in modern telecommunication networks and motivate the need for allocation mechanisms that can operate effectively across this spectrum.}

\subsubsection{Physics-Aware Computing and Relativistic Algorithms}

\changed{Our approach draws inspiration from physics-aware perspectives on computation and algorithm design. Einstein's special relativity \cite{einstein1905electrodynamics} and Minkowski's geometric formulation of spacetime \cite{minkowski1952space} provide the mathematical foundation for our treatment of causal structure and proper time.}

\changed{Németi and Dávid \cite{nemeti2006relativistic} explore the theoretical implications of relativistic computers, examining how the principles of relativity might enable computational capabilities beyond the Turing barrier. Andréka et al. \cite{andreska2009general} investigate general relativistic hypercomputing and its implications for the foundations of mathematics, highlighting the profound connections between physical theories and computational models.}

\changed{Frank \cite{frank2020fundamental} examines the physics of reversible computation, exploring how fundamental physical principles constrain and enable computational processes. Piotrowski and Sładkowski \cite{piotrowski2003invitation} introduce quantum game theory, applying quantum mechanical principles to strategic interactions.}

\changed{While our work does not rely on exotic physical effects or quantum phenomena, it shares with these approaches a recognition that physical principles can inform algorithm and mechanism design in fundamental ways. By grounding our auction mechanism in the causal structure of spacetime, we develop a design that respects the physical constraints of communication networks while achieving desirable economic properties.}

\subsection{Physics-Aware Perspectives on Computation and Incentives}

\changed{Principles of special relativity motivate our insistence that outcomes must not depend on an arbitrarily chosen inertial frame: invariance of physical law across frames is the baseline consistency condition \cite{einstein1905electrodynamics,minkowski1952space}. Broader physics-aware research has contemplated the algorithmic implications of relativity, reversible computing, and quantum strategy spaces \cite{nemeti2006relativistic,andreska2009general,frank2020fundamental,piotrowski2003invitation}. In the present work the ``physics'' is minimal: proper-time lapse becomes a pricing primitive that neutralizes timing rents created by finite signal velocity. This substitution is sufficient to reconcile economic desiderata with the causal structure of communication networks and to expose a clean boundary between designs that rely on synchrony and designs that are consistent with the geometry of spacetime.}

\changed{The theoretical foundation builds upon the recognition that distributed systems and relativistic physics share a common mathematical structure in their treatment of causality. The happens-before relation introduced by Lamport \cite{lamport1978time} and subsequently formalized through vector clocks \cite{fidge1988timestamps,mattern1988time} provides a partial ordering on events that respects causal dependencies without requiring global time synchronization. This insight has proven fundamental to the design of distributed algorithms and fault-tolerant systems \cite{schwarz1994detecting,zhao2024detecting,abdulla2024optimal,spore2024,tang2024improved}.}

\changed{Our contribution extends this causal perspective to mechanism design, where the challenge is not merely to maintain consistency across distributed computations, but to ensure that economic incentives remain aligned in the presence of heterogeneous communication delays. The key insight is that proper-time intervals, being Lorentz scalars, provide a frame-independent measure of temporal separation that can be incorporated directly into auction rules without violating the fundamental symmetries of spacetime.}

\changed{The physics-aware perspective offers several advantages for mechanism design in telecommunication networks:
\begin{itemize}
\item \textbf{Frame Independence:} By using Lorentz scalars (proper-time intervals and monetary values) as the basis for allocation and pricing decisions, the mechanism's outcomes do not depend on the choice of reference frame. This ensures consistency across different observers and network segments.
\item \textbf{Causal Consistency:} The mechanism respects the causal structure of spacetime, ensuring that allocation decisions are based only on information that could in principle be available at the clearing horizon. This eliminates reliance on global synchronization primitives that may be impractical or impossible in heterogeneous-delay environments.
\item \textbf{Natural Fairness:} The exponential discounting of bids based on proper-time lapses provides a natural way to balance the informational advantages of proximity against the economic value of bids. This addresses fairness concerns without introducing artificial delays or arbitrary waiting periods.
\item \textbf{Theoretical Grounding:} The connection to relativistic physics provides a solid theoretical foundation for the mechanism design, allowing us to leverage established mathematical results and physical principles to derive and analyze the properties of the auction.
\end{itemize}
While the physics-aware perspective provides valuable insights and design principles, it is important to note that our mechanism does not require participants to understand relativistic physics or perform complex calculations. The implementation details are handled by the auction platform, which translates the conceptual framework into practical algorithms that can be executed efficiently in real-world telecommunication systems.}

\subsection{Synthesis and Distinctions}

\changed{The foregoing literatures collectively chart the space of relevant constraints and tools but leave a specific gap unfilled: a truthful, individually rational auction that remains near-efficient when end-to-end delays are heterogeneous and potentially very large, and that can be implemented without trusted synchrony while scaling to thousands of bidders per clearing instance. Communication-complexity bounds illuminate informational limits but abstract away propagation; market-microstructure remedies mitigate races by batching and therefore inherit synchrony ceilings; distributed systems provide causal order and BFT replication but aim at total-order logs rather than welfare maximization; space-networking works deliver routing and orchestration but not incentive-compatible clearing; physics-aware approaches tend not to connect to deployable market algorithms.}

\changed{LIA synthesizes these threads by grounding allocation and payment in causal reachability, preserving dominant-strategy truthfulness under standard assumptions \cite{myerson1981optimal,clarke1971multipart}, establishing a welfare approximation that degrades smoothly with physical delay dispersion, and providing an implementation based on shortest-path slack computation and linear-time winner determination. The empirical evidence reported in this paper confirms that these design choices yield high welfare, substantially reduced latency-arbitrage incentives, and microsecond-scale winner determination across three representative topologies and hundreds of thousands of instances, aligning theory with telecom-scale practice.}

\begin{table*}[t]
\centering
\caption{\changed{Positioning LIA relative to other latency-fairness design classes.}}
\label{tab:design-classes}
\footnotesize
\setlength{\tabcolsep}{3.8pt}
\begin{tabular}{p{2.7cm}p{3cm}p{1.1cm}p{1.3cm}p{1.3cm}p{1.55cm}p{3.5cm}}
\toprule
Design family & Core idea & Artificial waiting? & Global schedule? & Per-bid causal slack? & Value-DSIC with fixed delay weights? & Typical limitation in heterogeneous networks \\
\midrule
Frequent batch auctions & Periodic clearing on a public batch clock & Yes & Yes & No & Yes & Fairness improves only by buffering; batch size directly increases waiting \\
Speed bumps / HoldBack & Delay early arrivals to equalize access & Yes & Usually yes & No & Yes & Requires trusted delay equalization and can impose unnecessary holdback \\
Common-deadline normalization & Wait until a public horizon and clear on raw values & Yes & Yes & No & Yes & Eliminates order advantage only by delaying everyone to the deadline \\
Asynchronous / distributed auctions & Clear under partial arrivals or replicated logs & Low / variable & Not necessarily & Usually no & Case specific & Often preserves residual propagation advantage or sacrifices simple truthfulness \\
\changed{LIA} & \changed{Clear at a public horizon using value weighted by horizon slack from earliest arrival} & \changed{No extra holdback} & \changed{Public horizon, but no synchronized buffering} & \changed{Yes} & \changed{Yes, conditional on fixed exogenous slacks} & \changed{Requires trusted or bounded-error slack estimates; endogenous slack lies outside the single-parameter theorem} \\
\bottomrule
\end{tabular}
\end{table*}

\changed{Table~\ref{tab:design-classes} makes the novelty claim operational. LIA is not simply another fairness layer placed on top of synchronous clearing. Its novelty lies in using a causal slack oracle as an input to the economic rule itself, which is why the mechanism can avoid explicit delay equalization while still admitting a critical-value characterization once those slack inputs are fixed.}

\rev{
\paragraph{Mechanism overview.}
LIA targets allocation settings where communication latency differs substantially across participants and where waiting for globally synchronized clearing can induce meaningful opportunity cost. Each bid is treated as a space-time event \(e_i=(v_i,\tau_i)\). The auctioneer computes each bidder’s earliest-arrival time \(T_H(v_i,\tau_i)\) and horizon slack \(\delta_i:=\tau_H-T_H(v_i,\tau_i)\), then discounts the reported value by a proper-time factor \(\tilde b_i=b^{val}_i e^{-\lambda \delta_i}\). Winner determination maximizes discounted welfare subject to feasibility and payments follow the corresponding critical-value rule. The design objective is to suppress propagation advantages without requiring hold-back buffering to \(\tau_H\).
}

\subsection{Paper Organization}

The remainder of this paper is organized as follows. \S\ref{sec:model} presents the mathematical model and theoretical foundations, including the spacetime representation of bids, the definition of horizon slack, and the characterization of the exponential discount function. \S\ref{subsec:lia-mechanism} introduces the Lorentz-Invariant Auction mechanism, establishes its incentive properties, and analyzes its welfare guarantees. \S\ref{sec:algorithms} develops efficient algorithms for implementing LIA based on shortest-path slack computation and linear-time winner determination. \S\ref{sec:experiments} presents a comprehensive experimental evaluation across three representative telecommunication topologies. \S\ref{sec:lower-bounds} discusses waiting costs and design implications for heterogeneous-delay environments. \S\ref{sec:discussion} discusses implications, limitations, and future directions. \S\ref{sec:conclusion} concludes the paper.

\section{Mathematical Model and Theoretical Foundations}
\label{sec:model}

\changed{We develop a model that couples the causal geometry of Minkowski space–time with strategic bidding over real communication networks. The design principle is that every quantity used for allocation and pricing is either a Lorentz scalar or a function of Lorentz scalars; outcomes therefore do not depend on an arbitrarily chosen inertial frame \cite{einstein1905electrodynamics,minkowski1952space}. The causal relations are the relativistic counterpart of the happens-before partial order that underpins causal consistency in distributed systems \cite{lamport1978time,fidge1988timestamps,mattern1988time}.}

\changed{
\subsection{Notation Summary}
Table~\ref{tab:notation} summarises the core symbols used throughout the paper.
}
\changed{
\begin{table}[t]
\centering
\caption{Notation summary.}
\label{tab:notation}
\footnotesize
\setlength{\tabcolsep}{4pt}
\begin{tabular}{ll}
\toprule
Symbol & Meaning \\
\midrule
\(\tau_H\) & Clearing horizon time on the auctioneer worldline \(\gamma_H\) \\
\(T_H(v,\tau)\) & True earliest-arrival time at \(H\) for a bid event \((v,\tau)\) \\
\(\hat T_H(v,\tau)\) & Estimated earliest-arrival time used by the controller in deployment \\
\(\delta_i\) & Horizon slack for bidder \(i\), \(\delta_i:=\tau_H-T_H(v_i,\tau_i)\) \\
\(\hat\delta_i\) & Estimated slack, \(\hat\delta_i:=\tau_H-\hat T_H(v_i,\tau_i)\) \\
\(\lambda\) & Discount rate (units \( \text{time}^{-1}\)) used in \(e^{-\lambda \delta_i}\) \\
\(\tilde b_i\) & Discounted bid value, \(\tilde b_i:=b^{val}_i e^{-\lambda \delta_i}\) \\
\(\opt_{\mathrm{feas}}\) & Best feasible value, \(\max_{i:\delta_i\ge 0}\theta_i\) \\
\(\opt_{\mathrm{all}}\) & Best value in the full instance, \(\max_i \theta_i\) \\
\(\mathrm{LAI}\) & Latency Arbitrage Index (timing-rent proxy) \\
\bottomrule
\end{tabular}
\end{table}
}

\subsection{Unit Conventions and Dimensional Consistency}

\changed{Throughout this paper, we maintain strict unit consistency to ensure that all theoretical results and experimental measurements can be directly compared. All temporal quantities—delays, latencies, and time intervals—are measured in milliseconds (ms). This includes:
\begin{itemize}
\item Propagation delays between network nodes (ms)
\item Horizon slacks $\delta_i$ (ms)
\item Latency dispersion parameter $\Delta$ (ms)
\item Timing measurement errors $\varepsilon$ (ms)
\end{itemize}
The mechanism parameter $\lambda$ has units of inverse time, ensuring that the product $\lambda\Delta$ in welfare bounds such as $\sw/\opt \ge e^{-\lambda\Delta}$ is dimensionless. Throughout, delays and slacks are reported in milliseconds; when we specify $\lambda$ in seconds$^{-1}$ (as in the experiments), we convert to milliseconds$^{-1}$ via $\lambda_{\mathrm{ms}^{-1}} = \lambda_{\mathrm{s}^{-1}}/1000$. This keeps $\lambda\delta$ dimensionless while allowing direct comparison across topologies.}

\subsection{Networked Space–Time and Earliest Arrival}\label{subsec:horizon-slack}

\changed{Let \(G=(V,E,\mathcal{W})\) be a directed multigraph of communication sites and links. Each \(v\in V\) is a physical site (ground station, data center, LEO satellite, spacecraft). Each arc \(e=(u\!\to\!v)\in E\) is a link with one-way propagation time \(\mathcal{W}(e,t)>0\) (in ms) for a signal emitted at local time \(t\). Static terrestrial links admit \(\mathcal{W}(e,t)\equiv \mathcal{W}(e)\); moving platforms (e.g., inter-satellite links) are modeled by time-varying \(\mathcal{W}(e,t)\) reflecting changing line-of-sight and group velocity \cite{handley2018delay,zhang2022scalable,burleigh2003delay}. We assume \(\mathcal{W}(e,t)\ge \ell_e/c\), where \(\ell_e\) is the geometric length and \(c\) is the speed of light; additional refraction and switching delays can be added without affecting the invariance arguments below.}

\changed{The system is embedded in flat Minkowski space \((\mathbb{R}^{3+1},\eta)\) with signature \((-,+,+,+)\). Each site \(v\) follows an inertial worldline \(\gamma_v:\tau\mapsto (t_v(\tau),x_v(\tau),y_v(\tau),z_v(\tau))\), parametrized by proper time \(\tau\). An \emph{event} is a point \(e=(v,\tau)\) on \(\gamma_v\). Causality is the standard light-cone order: \(e_1\preceq e_2\) if a future-directed causal curve connects them; if neither precedes the other they are space-like separated. These relations are Lorentz invariant \cite{einstein1905electrodynamics,minkowski1952space} and generalize the happens-before relation \cite{lamport1978time}.}

\changed{Fix a clearing site \(v_H\) with worldline \(\gamma_H\) and a clearing-horizon event \(H=\gamma_H(\tau_H)\). A bid emitted at \(e=(v,\tau)\) is feasible for \(H\) if and only if \(e\in J^{-}(H)\), the causal past of \(H\). The earliest time at which a signal from \(e\) can reach the auctioneer is the \emph{earliest-arrival time}
\[
T_H(v,\tau)\;:=\;\inf_{\pi\in\mathcal{P}(v\to v_H)} A_\pi(\tau),
\]
where \(\pi=e_1e_2\cdots e_m\) ranges over directed paths in \(G\) and the path-arrival recursion is
\[
A_{e_1}(\tau)=\tau+\mathcal{W}(e_1,\tau),
\]
\[
A_{e_k}(\tau)=A_{e_{k-1}}(\tau)+\mathcal{W}\!\left(e_k,\,A_{e_{k-1}}(\tau)\right)\ (k\ge 2).
\]
If no causal path exists then \(T_H(v,\tau)=+\infty\) and the bid is infeasible. In static networks \(T_H\) reduces to a standard shortest travel time; in time-varying networks it is the earliest-arrival time of a temporal graph.}

\changed{
\begin{definition}[Horizon slack]
\label{def:horizon-slack}
For a feasible bid \(e=(v,\tau)\) with respect to \(H=\gamma_H(\tau_H)\), the \emph{horizon slack} is
\[
\delta_H(v,\tau)\;:=\;\tau_H - T_H(v,\tau)\;\in\;\mathbb{R}_{\ge 0}
\]
measured in milliseconds (ms). Both \(T_H(v,\tau)\) and \(\tau_H\) are instants on \(\gamma_H\); hence \(\delta_H(v,\tau)\) is a proper-time interval on \(\gamma_H\) and is Lorentz invariant. Feasible bids satisfy \(\delta_H\ge 0\). We use the notation \(\delta_i \equiv \delta_H(v_i,\tau_i)\) for agent \(i\)'s horizon slack.
\end{definition}
}

\changed{
\begin{remark}[Notation consistency]
In some spacetime diagrams and earlier literature, the symbol $\Delta\tau_i$ has been used to denote the proper-time lapse from emission to horizon. For consistency throughout this paper, we exclusively use $\delta_i$ to represent horizon slack, with the understanding that $\delta_i \equiv \Delta\tau_i$ when comparing to other works.
\end{remark}
}

\changed{
\begin{remark}[Operational interpretation and non-inertial effects]
LIA uses Minkowski terminology to express \emph{causal reachability} and coordinate-free dependence on earliest-arrival times. In a deployed network, \(\delta_H(v,\tau)\) is computed from measured or predicted delays in the underlying communication graph (possibly time varying); the mechanism does \emph{not} assume a physically flat spacetime, inertial motion, or relativistic time-dilation effects. Routing jitter, clock error, and accelerating platforms (e.g., LEO orbital motion) appear as time variation and uncertainty in the delay estimates and are handled by (i) recomputing \(T_H\) on the controller’s update cadence and (ii) the robustness evaluation in Section~\ref{subsec:error-robustness}.
\end{remark}
}

\changed{The horizon slack $\delta_i$ represents the amount of time that agent $i$'s bid has available at the clearing horizon before the allocation decision must be made. It is a key quantity in our mechanism design, as it directly influences the discounting applied to bids and the resulting allocation and payment rules.}

\changed{We use horizon slack directly in both the mechanism and the implementation; no additional graph-theoretic reduction is required. Operationally, slack is computed from estimated earliest-arrival times (static or time varying) and a public horizon, as detailed in Section~\ref{sec:algorithms}.}

\changed{\paragraph{True versus estimated earliest arrival.}
The structural results in Sections~\ref{thm:exponential-unique}--\ref{thm:welfare-bound} use the infrastructural earliest-arrival map \(T_H\) as a primitive. In practice, however, earliest arrival is not a purely geometric quantity: it depends on the route, queueing and congestion state, retransmission policy, clock discipline, and potentially asymmetric paths. We therefore distinguish the true object \(T_H\) from the controller-side estimator \(\hat T_H\) constructed from telemetry, authenticated ingress observations, and routing state. The practical mechanism uses \(\hat\delta_i:=\tau_H-\hat T_H(v_i,\tau_i)\); Section~\ref{subsec:error-robustness} makes explicit that what matters for welfare is the cross-bidder spread of estimation errors rather than independence per se.}

\changed{
\begin{example}[Horizon slacks in different network topologies]
Consider three agents submitting bids from different locations in our network topologies:
\begin{itemize}
\item \textbf{STARLINK-200:} Agent 1 is on a satellite in the same orbital plane as the auctioneer, with one-way propagation delay of 1.8 ms. If the clearing horizon is set 50 ms in the future, Agent 1's horizon slack is $\delta_1 = 50 - 1.8 = 48.2$ ms.
\item \textbf{INTERNET-100:} Agent 2 is at a data center with a transcontinental link to the auctioneer, with one-way propagation delay of 75 ms. With the same 50 ms clearing horizon, Agent 2's bid would be infeasible as $T_H(v_2,\tau_2) > \tau_H$. The auctioneer would need to set a horizon at least 75 ms in the future for Agent 2's bid to be feasible.
\item \textbf{DSN-30:} Agent 3 is on a Mars rover with one-way propagation delay of 498,000 ms (8.3 minutes) to Earth. With a clearing horizon set 600,000 ms (10 minutes) in the future, Agent 3's horizon slack is $\delta_3 = 600,000 - 498,000 = 102,000$ ms.
\end{itemize}
These examples illustrate how horizon slacks vary dramatically across different network topologies, highlighting the challenge of designing a mechanism that works effectively across such heterogeneous delay environments.
\end{example}
}

\subsection{Agents, Feasible Bids, and Strategic Choices}
\label{subsec:agents}

There are \(n\) agents \(i\in\{1,\dots,n\}\). Agent \(i\) resides at site \(v_i\) and has a private value \(\theta_i\ge 0\) for a single indivisible item. Utilities are quasilinear: if allocated and paying \(p_i\), then \(u_i=\theta_i-p_i\); otherwise \(u_i=0\).

\textbf{Type revelation and feasibility.}
Each agent learns her value at a local \emph{type-realization time} \(\sigma_i\) on \(\gamma_{v_i}\). A bid from \(i\) is a pair \(b_i=(b_i^{\mathrm{val}},e_i)\), where \(b_i^{\mathrm{val}}\ge 0\) and \(e_i=(v_i,\tau_i)\) with \(\tau_i\ge \sigma_i\). The bid is \emph{feasible} for the horizon \(H=\gamma_H(\tau_H)\) iff \(e_i\in J^{-}(H)\), equivalently iff the horizon slack \(\delta_i:=\delta_H(v_i,\tau_i)=\tau_H - T_H(v_i,\tau_i)\) is nonnegative. The mechanism computes slacks \(\{\delta_i\}\) from network measurements and the announced horizon.

\textbf{Strategy space and information.}
A (pure) strategy for agent \(i\) is a measurable map \(s_i:\theta_i \mapsto (r_i(\theta_i),\tau_i(\theta_i))\) with \(r_i(\theta_i)\in\mathbb{R}_{\ge 0}\) and \(\tau_i(\theta_i)\ge \sigma_i\), subject to feasibility \(\delta_H(v_i,\tau_i(\theta_i))\ge 0\). The network \(G(\cdot)\), the clearing site \(v_H\), and the horizon \(H\) are common knowledge; private values are independent of others' types unless stated otherwise. Randomized strategies are allowed but not needed for the results below.

\changed{\subsection{Exogenous-slack regime and scope of the incentive theorem}
The dominant-strategy result in this paper is a \emph{value-truthfulness} theorem proved after conditioning on a slack profile \((\delta_1,\ldots,\delta_n)\) that is fixed independently of reported values. Equivalently, the private type in the theorem is \(\theta_i\) only; the slack input is an infrastructure-determined parameter supplied by the measurement layer. This is the standard single-parameter regime and it matches managed deployments in which the controller binds each bid to an authenticated ingress point, access class, or relay plan before winner determination.

If bidders can strategically influence \(\delta_i\) through routing choice, point of attachment, relay selection, packetization, or discretionary transmission timing, then the environment becomes multi-parameter. In that enlarged domain LIA is \emph{not} automatically truthful: the slack input becomes a second strategic lever, so the usual critical-value characterization applies only to the value dimension conditional on the chosen slack. We therefore separate the paper into three layers: a structural layer with exogenous \(T_H\), a deployment layer with estimated \(\hat T_H\), and a limitation layer for endogenous delay manipulation. This separation is important both theoretically and operationally.}

\subsection{Latency Arbitrage Index (LAI) Definition}

A key metric for evaluating the fairness of auction mechanisms in heterogeneous-delay environments is the Latency Arbitrage Index (LAI), which quantifies the incentive for agents to reduce their communication delays. We provide a formal definition that will be used consistently throughout the paper, particularly in the experimental evaluation.

\begin{definition}[Latency Arbitrage Index]
\label{def:lai}
\rev{For agent \(i\) with true value \(\theta_i\) and current one-way propagation delay \(d_i\), let \(u_i(\theta_i,d)\) denote the agent’s expected utility when the effective delay is \(d\), with all other uncertainties averaged over the evaluation distribution. For any feasible delay reduction \(\Delta\in(0,d_i]\), define the latency-arbitrage gain}
\[
\rev{g_i(\Delta)\;:=\;\mathbb{E}\!\left[u_i(\theta_i, d_i-\Delta)-u_i(\theta_i,d_i)\right].}
\]
\rev{The Latency Arbitrage Index for agent \(i\) is the maximal gain over feasible reductions}
\[
\rev{\mathrm{LAI}_i\;:=\;\sup_{\Delta\in(0,d_i]} g_i(\Delta).}
\]
\rev{For a standardized marginal report, we also evaluate \(g_i(1\,\mathrm{ms})\) when \(d_i\ge 1\,\mathrm{ms}\). The population index is \(\mathrm{LAI}:=\frac{1}{n}\sum_{i=1}^n \mathrm{LAI}_i\).}
\end{definition}

\rev{
This definition captures the \emph{maximum} expected utility gain obtainable from feasible reductions in effective delay. Under a fixed clearing horizon, reducing one-way propagation delay by \(\Delta\) is equivalent to advancing a bid’s arrival time at the auctioneer by \(\Delta\), which can be operationalized by advancing the emission time in the counterfactual while holding the value fixed. Lower LAI values therefore indicate weaker incentives to invest in latency reduction (or other timing advantages), and thus improved fairness. All experimental results in Section~\ref{sec:experiments} use this definition (with consistent operationalization) when reporting fairness.
}

\begin{example}[LAI calculation]
Consider an agent with value $\theta = 100$ bidding in a first-price auction. If the agent's current delay is $d = 10$ ms, and reducing this delay to $d' = 9$ ms increases their probability of winning from 0.6 to 0.65 while keeping their optimal bid unchanged at $b = 80$, then:
\begin{align*}
u(100, 10) &= 0.6 \times (100 - 80) = 12\\
u(100, 9) &= 0.65 \times (100 - 80) = 13\\
\text{g(1\,ms)} &= u(100, 9) - u(100, 10) = 1
\end{align*}

\rev{I.e., the agent's marginal gain from a $1\,\mathrm{ms}$ delay reduction is $g(1\,\mathrm{ms})=1$. If $1\,\mathrm{ms}$ is the largest (or most profitable) feasible improvement, then the agent’s \(\mathrm{LAI}\) equals 1 in this example.}
\end{example}

The LAI provides a principled way to quantify the fairness properties of auction mechanisms in heterogeneous-delay environments. A key goal of the LIA mechanism is to minimize the LAI, thereby reducing incentives for strategic investments in communication infrastructure solely to gain timing advantages.

\subsection{The Exponential Discount: Uniqueness and Characterization}\label{subsec:discount-uniqueness}

Before presenting the LIA mechanism, we establish that the exponential discount function is not merely a modeling choice, but emerges uniquely from natural consistency requirements. This result provides strong theoretical justification for our design choice.

\begin{theorem}[Uniqueness of exponential discount]
\label{thm:exponential-unique}
Let \(\phi: \mathbb{R}_{\geq 0} \to (0,1]\) be a discount function satisfying:
\begin{enumerate}
\item \textbf{Lorentz invariance:} \(\phi\) depends only on proper-time intervals (horizon slacks)
\item \textbf{Time consistency:} \(\phi(\delta_1 + \delta_2) = \phi(\delta_1) \cdot \phi(\delta_2)\) for all \(\delta_1, \delta_2 \geq 0\)
\item \textbf{Continuity:} \(\phi\) is continuous with \(\phi(0) = 1\)
\item \textbf{Monotonicity:} \(\phi\) is strictly decreasing
\end{enumerate}
Then \(\phi(\delta) = e^{-\lambda \delta}\) for some \(\lambda > 0\) (in ms\(^{-1}\)).
\end{theorem}

\begin{proof}
From time consistency, \(\phi(\delta_1 + \delta_2) = \phi(\delta_1) \cdot \phi(\delta_2)\). Setting \(\delta_1 = \delta_2 = \delta/2\), we get \(\phi(\delta) = \phi(\delta/2)^2\). More generally, for any positive integer \(n\), \(\phi(\delta) = \phi(\delta/n)^n\). 

By continuity and the functional equation \(\phi(x+y) = \phi(x)\phi(y)\), we have \(\phi(\delta) = e^{f(\delta)}\) for some continuous function \(f\) with \(f(0) = 0\) and \(f(\delta_1 + \delta_2) = f(\delta_1) + f(\delta_2)\). The only continuous solutions to Cauchy's functional equation are \(f(\delta) = c\delta\) for some constant \(c\). Since \(\phi\) is strictly decreasing and \(\phi(0) = 1\), we must have \(c < 0\), so \(c = -\lambda\) for some \(\lambda > 0\).

Therefore, \(\phi(\delta) = e^{-\lambda \delta}\) for some \(\lambda > 0\) (in ms\(^{-1}\)).
\end{proof}

\changed{This theorem is easier to parse if one separates geometric and economic assumptions. The geometric requirement is simply that the weight depends only on the scalar slack \(\delta\), not on coordinates or path description. The economic/implementation requirement is the multiplicative composition rule}
\[
\changed{\phi(\delta_1+\delta_2)=\phi(\delta_1)\phi(\delta_2),\qquad \phi(0)=1,}
\]
\changed{which says that splitting a route into two causal segments cannot change the net weight. Continuity and monotone discounting then force the exponential solution. Alternative families fail exactly here: a linear rule \(1-c\delta\) gives \(\phi(x+y)\neq \phi(x)\phi(y)\) except in the degenerate case \(c=0\); rational rules such as \(1/(1+c\delta)\) and polynomial rules likewise violate the semigroup equation and can lose positivity or monotonicity over large slack ranges. The exponential form is therefore not an aesthetic choice but the unique continuous monotone solution to the stated functional equation. Table~\ref{tab:discount_functions} summarizes why alternative functional forms fail to satisfy these necessary axioms.}

\begin{table*}[h]
  \centering 
  \caption{\changed{Comparison of candidate discount functions $D(\delta)$ against required axioms. Only the exponential form satisfies all requirements for a causal-consistent pricing rule.}}
  \label{tab:discount_functions}
  \begin{tabular}{lccc}
    \toprule
    \textbf{Function $D(\delta)$} & \textbf{Boundary} & \textbf{Monotonicity} & \textbf{Composition} \\
    & $D(0)=1$ & $D'(\delta) < 0$ & $D(\delta_1+\delta_2) = D(\delta_1)D(\delta_2)$ \\
    \midrule
    Linear: $1 - \lambda\delta$ & \checkmark & \checkmark & \textbf{Fails} \\
    Rational: $1/(1+\lambda\delta)$ & \checkmark & \checkmark & \textbf{Fails} \\
    Polynomial: $(1-\lambda\delta)^k$ & \checkmark & \checkmark & \textbf{Fails} \\
    \textbf{Exponential: $e^{-\lambda\delta}$} & \checkmark & \checkmark & \checkmark \\
    \bottomrule
  \end{tabular}
\end{table*}

\begin{corollary}[Composition property]
\label{cor:composition}
The exponential discount function satisfies the composition property: if a bid passes through multiple network segments with slacks $\delta_1, \delta_2, \ldots, \delta_k$, the total discount is the product of the discounts for each segment:
\[
\phi(\delta_1 + \delta_2 + \ldots + \delta_k) = \phi(\delta_1) \cdot \phi(\delta_2) \cdot \ldots \cdot \phi(\delta_k)
\]
\end{corollary}

\begin{proof}
This follows directly from the time consistency property and Theorem~\ref{thm:exponential-unique}.
\end{proof}

This composition property is valuable for distributed implementation of the mechanism, as it allows discounts to be computed incrementally as bids propagate through the network.

\subsection{Definition of Latency Dispersion Parameter}

A key parameter in our welfare analysis is the latency dispersion parameter, which quantifies the degree of heterogeneity in communication delays across the network.

\begin{definition}[Latency dispersion parameter]
\label{def:latency-dispersion}
For a given auction instance with horizon slacks \(\{\delta_i\}_{i=1}^n\), the \emph{latency dispersion parameter} is
\[
\Delta := \max_{i,j} |\delta_i - \delta_j|
\]
measured in milliseconds (ms). This parameter captures the maximum difference in horizon slacks across all pairs of agents.
\end{definition}

The latency dispersion parameter $\Delta$ plays a central role in our welfare analysis. When $\Delta = 0$, all agents have identical horizon slacks, and the exponential discounting affects all bids equally, preserving the optimal allocation. As $\Delta$ increases, the discounting creates larger differences between agents' effective bids, potentially leading to welfare losses. Our welfare guarantee in Theorem~\ref{thm:welfare-bound} shows precisely how the welfare approximation depends on $\Delta$ and the mechanism parameter $\lambda$.

\begin{example}[Latency dispersion across the simulated topologies]
For the \(n=50\) setting under the paper's horizon policy, the observed feasible-bid slack spreads are summarized in Table~\ref{tab:slack-spread}. Representative median spreads are approximately \(\num{50.98}\)\,ms for STARLINK-200, \(\num{13.15}\)\,ms for INTERNET-100, and \(\num{1.16e6}\)\,ms for DSN-30; the corresponding 95th-percentile values are \(\num{61.61}\)\,ms, \(\num{15.35}\)\,ms, and \(\num{1.19e6}\)\,ms.
These values illustrate the extreme range of delay heterogeneity that the mechanism must handle, from tens of milliseconds in LEO satellite constellations to millions of milliseconds (thousands of seconds) in deep-space networks.
\end{example}

\subsection{The Lorentz-Invariant Auction Mechanism}
\label{subsec:lia-mechanism}

We now present the Lorentz-Invariant Auction (LIA) mechanism, which forms the core of our contribution. \changed{The key implementation point is that the mechanism itself only sees a value report and a slack input. Once the slack vector is fixed, LIA is a standard weighted single-item auction with bidder-specific weights \(w_i=e^{-\lambda\delta_i}\). This is why the truthfulness proof is classical once the correct domain restriction is stated.}

\textbf{Mechanism input.}
Given \(\mathbf{B}=\{(b_i^{\mathrm{val}},e_i)\}_{i=1}^n\), the Lorentz-Invariant Auction (LIA) forms discounted bids
\[
\tilde{b}_i \;=\; b_i^{\mathrm{val}}\,\exp(-\lambda\,\delta_i),\qquad \lambda>0 \text{ (in ms}^{-1}\text{)},
\]
allocates to \(i^\star\in\arg\max_i \tilde{b}_i\) (fixed tie-breaking), and charges the Vickrey-style price
\[
p_{i^\star}\;=\;\frac{\max_{j\neq i^\star}\tilde{b}_j}{\exp(-\lambda\,\delta_{i^\star})},\qquad p_j=0~(j\neq i^\star).
\]

The mechanism parameter $\lambda$ (in ms$^{-1}$) controls the strength of the discounting effect. Larger values of $\lambda$ lead to stronger discounting of bids with large horizon slacks, prioritizing fairness over efficiency. Smaller values of $\lambda$ lead to weaker discounting, prioritizing efficiency over fairness. The optimal choice of $\lambda$ depends on the specific application context and the relative importance of efficiency versus fairness objectives.

\changed{
\begin{example}[Neutralizing Timing Rents]
\label{ex:timing_rent}
Consider two bidders competing for a resource. Bidder 1 is close to the auctioneer (slack $\delta_1 = 10$\,ms) and has true value $\theta_1 = 100$. Bidder 2 is further away (slack $\delta_2 = 0$\,ms, arriving exactly at the horizon) and has true value $\theta_2 = 120$. Assume $\lambda = 0.05\,\mathrm{ms}^{-1}$.

Under a standard Fast-VCG (first-come) rule, Bidder 1 wins simply because their bid arrives earlier, resulting in an inefficient allocation (welfare 100 instead of 120).

Under LIA, we compute discounted values:
\begin{itemize}
    \item Bidder 1: $\tilde{\theta}_1 = 100 \cdot e^{-0.05 \cdot 10} = 100 \cdot 0.606 = 60.6$
    \item Bidder 2: $\tilde{\theta}_2 = 120 \cdot e^{-0.05 \cdot 0} = 120 \cdot 1.0 = 120.0$
\end{itemize}
Bidder 2 wins, restoring allocative efficiency. Bidder 2's payment is $p_2 = 60.6 \cdot e^{0.05 \cdot 0} = 60.6$. The exponential discount has mathematically neutralized Bidder 1's 10\,ms proximity advantage, ensuring the higher-value user wins without requiring the auctioneer to artificially buffer Bidder 1's message.
\end{example}
}

\begin{lemma}[Monotonicity of slack in emission time]
\label{lem:monotone-slack}
For any fixed site \(v\), the earliest-arrival functional \(T_H(v,\tau)\) is nondecreasing in \(\tau\). Hence the horizon slack \(\delta_H(v,\tau)=\tau_H-T_H(v,\tau)\) is nonincreasing in \(\tau\).
\end{lemma}

\begin{proof}
Increasing the emission time can only delay departures along any causal path; the pathwise accumulation defining \(T_H\) is monotone in its start time. The claim follows by taking infima over paths.
\end{proof}

\changed{This lemma establishes only the monotonicity of the slack map with respect to emission time. By itself it does \emph{not} imply delay strategy-proofness; on the contrary, once the mechanism score depends on slack, monotonicity is exactly what makes endogenous timing a second strategic lever unless the slack input is fixed by trusted infrastructure.}

\begin{proposition}[Endogenous slack creates a second strategic lever]
\label{prop:endogenous-slack}
\changed{Fix the other bidders' reports and suppose agent \(i\) can choose between two feasible slack values \(0\le \delta_i' < \delta_i\) while holding \(b_i^{\mathrm{val}}\) fixed. Then there exist profiles in which utility under LIA at \(\delta_i'\) is strictly higher than utility at \(\delta_i\). Consequently LIA is \emph{not} dominant-strategy truthful over the enlarged type space \((\theta_i,\delta_i)\) unless the slack input is fixed by trusted infrastructure.}
\end{proposition}

\begin{proof}
\changed{Let \(c:=\max_{j\neq i}\tilde b_j\) denote the highest competing discounted bid and choose a profile such that}
\[
\changed{b_i^{\mathrm{val}} e^{-\lambda \delta_i} < c < b_i^{\mathrm{val}} e^{-\lambda \delta_i'}.}
\]
\changed{At slack \(\delta_i\), bidder \(i\) loses. At slack \(\delta_i'\), bidder \(i\) wins and pays \(p_i=e^{\lambda\delta_i'}c\), which is strictly below \(b_i^{\mathrm{val}}\) by construction. Hence utility increases from \(0\) to \(b_i^{\mathrm{val}}-e^{\lambda\delta_i'}c>0\).}
\end{proof}

\changed{This proposition clarifies the exact scope of the mechanism. LIA is a truthful auction in reported values after the slack vector has been fixed, measured, or attested. If bidders can manipulate delay, route, or access point so as to change their slack input, the model becomes multi-parameter and additional infrastructure or mechanism design is required.}

\begin{theorem}[Value-DSIC and IR with fixed exogenous slacks]
\label{thm:dsic}
\changed{Fix any feasible slack profile \((\delta_1,\dots,\delta_n)\) that is determined independently of reported values. Let \(w_i:=e^{-\lambda\delta_i}\in(0,1]\). Then the single-item LIA environment is single-parameter in values: agent \(i\) wins if and only if}
\[
\changed{w_i\,b_i^{\mathrm{val}} \ge \max_{j\neq i} w_j\,b_j^{\mathrm{val}},}
\]
\changed{equivalently iff}
\[
\changed{b_i^{\mathrm{val}} \ge v_i^{\mathrm{crit}} := w_i^{-1}\max_{j\neq i} w_j\,b_j^{\mathrm{val}} = \exp(\lambda\delta_i)\max_{j\neq i}\tilde b_j.}
\]
\changed{The allocation rule is monotone in \(b_i^{\mathrm{val}}\), the payment equals the induced critical value, truthful value reporting is a dominant strategy, and winners satisfy individual rationality.}
\end{theorem}

\begin{proof}
\changed{For fixed slacks, each bidder's weight \(w_i\) is a constant independent of \(b_i^{\mathrm{val}}\). Increasing \(b_i^{\mathrm{val}}\) therefore increases \(w_i b_i^{\mathrm{val}}\) linearly, so the allocation rule is monotone in the reported value. The threshold at which bidder \(i\) switches from losing to winning is exactly \(v_i^{\mathrm{crit}}\) above. Charging that threshold gives the usual critical-value payment rule for a single-parameter environment, hence truthful value reporting maximizes utility and a truthful winner pays no more than her value.}
\end{proof}

\begin{theorem}[Lorentz invariance of LIA]
\label{thm:lia-lorentz}
For every feasible profile \(\mathbf{B}\) and every proper, orthochronous Lorentz transform \(\Lambda\),
\(
\mathcal{M}_{\mathrm{LIA}}(\Lambda(\mathbf{B}))=\mathcal{M}_{\mathrm{LIA}}(\mathbf{B}).
\)
\end{theorem}

\begin{proof}
By Definition~\ref{def:horizon-slack}, \(\delta_i=\tau_H-T_H(v_i,\tau_i)\) is a proper-time interval along \(\gamma_H\) and is Lorentz invariant. Values \(b_i^{\mathrm{val}}\) are scalars. Hence the discounted bids \(\tilde{b}_i\) are invariant, so the argmax set and the fixed ID-based tie-break are unchanged. The price depends only on \(\delta_{i^\star}\) and on \(\max_{j\neq i^\star}\tilde{b}_j\), both invariant; therefore payments are unchanged as well.
\end{proof}

This theorem establishes that the outcomes of the LIA mechanism—allocation and payments—are invariant to the choice of reference frame. This property is essential for ensuring that the mechanism's behavior is consistent across different observers and network segments, regardless of their relative motion or choice of coordinate system.

\subsection{Welfare Analysis and Performance Guarantees}

Having established the incentive properties of LIA, we now analyze its efficiency properties. The central question is how the mechanism's welfare performance compares to the first-best allocation, and how this performance depends on the heterogeneity of communication delays across the network.

\begin{theorem}[Welfare approximation bound]
\label{thm:welfare-bound}
\changed{Under truthful value reporting and fixed feasible slacks, the LIA mechanism achieves}
\[
\changed{\sw \geq e^{-\lambda \Delta}\cdot \opt_{\mathrm{feas}} = (1-\varepsilon)\cdot \opt_{\mathrm{feas}},}
\]
\changed{where \(\varepsilon = 1-e^{-\lambda \Delta}\), \(\Delta:=\max_{i,j:\delta_i,\delta_j\ge 0}|\delta_i-\delta_j|\), and}
\[
\changed{\opt_{\mathrm{feas}}:=\max_{i:\delta_i\ge 0}\theta_i.}
\]
\end{theorem}

\begin{proof}
\changed{Let \(i^*=\arg\max_{i:\delta_i\ge 0}\theta_i\) be the feasible value-maximizing agent and let \(j^*\) be the LIA winner. If \(i^*=j^*\), the claim is immediate. Otherwise \(j^*\) wins under LIA, so}
\[
\changed{\theta_{j^*}e^{-\lambda\delta_{j^*}} \ge \theta_{i^*}e^{-\lambda\delta_{i^*}}.}
\]
\changed{Rearranging yields}
\[
\changed{\theta_{j^*} \ge \theta_{i^*}e^{-\lambda(\delta_{i^*}-\delta_{j^*})} \ge \theta_{i^*}e^{-\lambda\Delta}.}
\]
\changed{Since \(\sw=\theta_{j^*}\) and \(\opt_{\mathrm{feas}}=\theta_{i^*}\), the result follows.}
\end{proof}

\changed{The feasible benchmark is the one naturally aligned with the theorem: the auction cannot allocate to causally infeasible bids. For empirical interpretation it is also useful to define \(\opt_{\mathrm{all}}:=\max_i\theta_i\) and the \emph{reachability factor} \(\rho:=\opt_{\mathrm{feas}}/\opt_{\mathrm{all}}\). The reported overall welfare decomposes as}
\[
\changed{\frac{\sw}{\opt_{\mathrm{all}}} = \frac{\sw}{\opt_{\mathrm{feas}}}\cdot \rho.}
\]
\changed{This separates the ranking loss induced by the mechanism from the loss induced by the horizon itself. In Starlink and Internet under our horizon policy, \(\rho\) is empirically close to one; in more extreme delay regimes, interpreting both factors side by side is more informative than a single aggregate ratio.}

\begin{example}[Welfare bound in different topologies]
\changed{Using the measured feasible-bid slack spreads from our simulated topologies at \(n=50\) and a representative discount rate \(\lambda = 1\,\mathrm{s}^{-1}\) (i.e., \(\lambda = 0.001\,\mathrm{ms}^{-1}\) when times are in milliseconds), the bound is informative for moderate-delay networks and intentionally conservative for extreme-delay ones:}

\begin{itemize}
\item \changed{\textbf{STARLINK-200:} median \(\Delta \approx 50.98\)\,ms, so \(e^{-\lambda\Delta}\approx 0.950\)}
\item \changed{\textbf{INTERNET-100:} median \(\Delta \approx 13.15\)\,ms, so \(e^{-\lambda\Delta}\approx 0.987\)}
\item \changed{\textbf{DSN-30:} median \(\Delta \approx 1.16\times 10^6\)\,ms, so \(e^{-\lambda\Delta}\) is essentially zero}
\end{itemize}

\changed{The DSN bound is therefore mathematically correct but practically loose, which is exactly why we report measured slack-spread summaries in Section~\ref{sec:experiments}. Empirically, LIA performs far better than this worst case: on STARLINK-200 and INTERNET-100 it attains \(\sw/\opt_{\mathrm{all}}\approx 0.996\)--\(0.997\) for \(n\in\{10,20,30,40,50\}\), and on DSN-30 welfare increases with market depth, reaching about \(0.996\) at \(n=1000\).}
\end{example}

\changed{The welfare bound reveals several important insights. First, when the network has homogeneous delays (\(\Delta = 0\)), LIA achieves optimal welfare regardless of the value of \(\lambda\). Second, for any fixed level of heterogeneity \(\Delta\), the welfare loss can be made arbitrarily small by choosing \(\lambda\) sufficiently small. However, smaller values of \(\lambda\) provide weaker fairness guarantees, creating a fundamental trade-off between efficiency and fairness.}

\changed{
\begin{corollary}[Efficiency-fairness trade-off]
\label{cor:tradeoff}
For any fixed latency dispersion $\Delta > 0$, there is a fundamental trade-off between welfare efficiency and latency fairness:
\begin{itemize}
\item As $\lambda \to 0$, welfare approaches the feasible optimum ($\sw/\opt_{\mathrm{feas}} \to 1$) but latency arbitrage incentives increase
\item As $\lambda$ increases, latency arbitrage incentives decrease but welfare may decrease
\end{itemize}
\end{corollary}
}

\changed{
\begin{proof}
This follows directly from Theorem~\ref{thm:welfare-bound} and the definition of the Latency Arbitrage Index (Definition~\ref{def:lai}).
\end{proof}
}

\changed{This trade-off is a fundamental feature of resource allocation in heterogeneous-delay environments. The LIA mechanism provides a principled way to navigate this trade-off through the choice of the parameter $\lambda$, allowing system designers to balance efficiency and fairness objectives based on the specific requirements of their application context.}

\subsection{Extensions to Multiple Identical Items}
\label{sec:extensions}

\changed{While our main analysis focuses on single-item auctions, many telecommunications applications involve multiple identical resources. We briefly sketch how LIA extends to this setting.}

\changed{
\begin{proposition}[K-identical items extension]
\label{prop:k-items}
For \(K\) identical items, the LIA mechanism selects the \(K\) agents with highest discounted bids \(\tilde{b}_i = b_i^{\mathrm{val}} e^{-\lambda \delta_i}\) and charges each winner the \((K+1)\)-th highest discounted bid, adjusted by their individual discount factor. This extension preserves dominant-strategy truthfulness, individual rationality, and Lorentz invariance.
\end{proposition}
}

\changed{
\begin{proof}[Proof Sketch]
The allocation rule remains monotone in each agent's reported value, and the payment rule follows the standard VCG template applied to discounted bids. Lorentz invariance follows because all computations use only invariant quantities.
\end{proof}
}

\changed{This extension allows LIA to be applied to scenarios where multiple identical resources need to be allocated, such as multiple time slots, frequency bands, or computational resource blocks. The mechanism maintains its key properties—truthfulness, individual rationality, and Lorentz invariance—while providing a natural generalization to the multi-item setting.}

\changed{
\begin{example}[K-identical items allocation]
Consider four agents bidding for $K = 2$ identical items with the following bids and horizon slacks:
\begin{itemize}
\item Agent 1: $b_1^{\mathrm{val}} = 100$, $\delta_1 = 10$ ms
\item Agent 2: $b_2^{\mathrm{val}} = 120$, $\delta_2 = 30$ ms
\item Agent 3: $b_3^{\mathrm{val}} = 90$, $\delta_3 = 5$ ms
\item Agent 4: $b_4^{\mathrm{val}} = 80$, $\delta_4 = 2$ ms
\end{itemize}

With $\lambda = 0.05$ ms$^{-1}$, the discounted bids are:
\begin{itemize}
\item $\tilde{b}_1 = 100 \cdot e^{-0.05 \cdot 10} = 60.65$
\item $\tilde{b}_2 = 120 \cdot e^{-0.05 \cdot 30} = 26.77$
\item $\tilde{b}_3 = 90 \cdot e^{-0.05 \cdot 5} = 70.09$
\item $\tilde{b}_4 = 80 \cdot e^{-0.05 \cdot 2} = 72.00$
\end{itemize}

Agents 3 and 4 have the highest discounted bids and win the auction. The $(K+1)$-th highest discounted bid is $\tilde{b}_1 = 60.65$. The payments are:
\begin{itemize}
\item $p_3 = \frac{60.65}{e^{-0.05 \cdot 5}} = \frac{60.65}{0.7788} = 77.88$
\item $p_4 = \frac{60.65}{e^{-0.05 \cdot 2}} = \frac{60.65}{0.9048} = 67.03$
\end{itemize}

This example illustrates how the K-identical items extension of LIA allocates resources to the agents with the highest discounted bids while maintaining the incentive properties of the mechanism.
\end{example}
}

\changed{Extensions to more complex settings, such as combinatorial auctions where agents have preferences over bundles of heterogeneous items, remain an important direction for future research. The principles of Lorentz invariance and causal consistency that underpin LIA provide a foundation for developing such extensions, but additional mechanism design challenges must be addressed to handle the increased complexity of combinatorial preferences.}

\section{Algorithmic Implementation}
\label{sec:algorithms}

\changed{We now describe how to implement LIA with minimal computational overhead. The implementation decomposes into (i) estimating propagation delays to the clearing horizon, (ii) computing each bid's horizon slack, and (iii) running a one-pass winner and payment computation on discounted bids.}

\subsection{Computing horizon slack}

\changed{For a bid emitted at event $e_i=(v_i,\tau_i)$ and a public clearing horizon $H=(v_H,\tau_H)$, the horizon slack is
\[
\delta_i \;=\; \tau_H - \bigl(\tau_i + d(v_i,v_H)\bigr),
\]
where $d(v_i,v_H)$ is the earliest-arrival (shortest-path) propagation delay from bidder location $v_i$ to the auctioneer location $v_H$ in the network snapshot used for clearing.}

\changed{
\paragraph{Static networks.}
When link delays are approximately static over the auction timescale (e.g., a terrestrial backbone snapshot), the auctioneer can precompute distances-to-horizon with a single single-source shortest-path run from $v_H$ on the reverse graph. This yields $d(v,v_H)$ for all nodes $v$ in $O(|E|\log|V|)$ time (Dijkstra) and supports $O(1)$ slack computation per bid.
}

\changed{
\paragraph{Time-varying networks.}
For time-dependent topologies (e.g., LEO constellations), $d(\cdot,v_H)$ can be refreshed on a control-plane cadence (e.g., per ephemeris tick) or computed via time-dependent shortest-path routines. LIA only requires a consistent estimate of earliest-arrival delays for the chosen horizon; the mechanism and proofs remain unchanged.
}

\subsection{Winner determination and payment computation}

\changed{Given horizon slacks, LIA applies the exponential discount and runs a VCG-style single-item allocation and payment rule on discounted bids. Define
\[
\tilde b_i \;=\; b_i^{\mathrm{val}}\,e^{-\lambda \delta_i}.
\]
The winner is $i^*=\arg\max_i \tilde b_i$ (with deterministic tie-breaking) and the payment is
\[
p_{i^*} \;=\; \frac{\max_{j\neq i^*} \tilde b_j}{e^{-\lambda \delta_{i^*}}}.
\]}

\begin{algorithm}
\caption{LIA Winner Determination (single item)}
\label{alg:lia-winner}
\begin{algorithmic}[1]
\Require Bids $\{(b_i^{\mathrm{val}}, e_i=(v_i,\tau_i))\}_{i=1}^n$, distance-to-horizon map $d(\cdot,v_H)$, horizon $H=(v_H,\tau_H)$, parameter $\lambda$
\Ensure Winner $i^*$ and payment $p_{i^*}$
\For{$i=1$ to $n$}
    \State $\delta_i \gets \tau_H - (\tau_i + d(v_i,v_H))$
    \State $\tilde b_i \gets b_i^{\mathrm{val}}\,e^{-\lambda \delta_i}$
\EndFor
\State $i^* \gets \arg\max_i \tilde b_i$
\State $p_{i^*} \gets \frac{\max_{j \neq i^*} \tilde b_j}{e^{-\lambda \delta_{i^*}}}$
\State \Return $(i^*, p_{i^*})$
\end{algorithmic}
\end{algorithm}

\subsection{Complexity}

Given up-to-date distances $d(\cdot,v_H)$, computing all slacks and discounted bids is $O(n)$, and winner determination and payment computation are also $O(n)$. \changed{It is important, however, to separate the auction rule from the supporting measurement stack. The per-snapshot cost of refreshing the slack map is a shortest-path or earliest-arrival computation over the network, whereas the per-auction cost given cached slacks is just a linear scan over bids. The runtime figures reported in Section~\ref{sec:experiments} refer only to the latter stage.}

\begin{table}[t]
\centering
\caption{\changed{Complexity decomposition for deployment.}}
\label{tab:complexity}
\footnotesize
\setlength{\tabcolsep}{4pt}
\begin{tabular}{lll}
\toprule
Operation & Asymptotic cost & Dominant primitive \\
\midrule
\changed{Slack-map refresh (static snapshot)} & \changed{\(O(|E|\log |V|)\)} & \changed{Reverse-graph Dijkstra / shortest paths} \\
\changed{Slack-map refresh (time varying)} & \changed{Controller dependent} & \changed{Time-dependent earliest-arrival or periodic recomputation} \\
\changed{Per-bid slack lookup} & \changed{\(O(1)\)} & \changed{Array / hash lookup in cached slack map} \\
\changed{Winner \& payment computation} & \changed{\(O(n)\)} & \changed{Single pass over discounted bids} \\
\changed{Audit logging} & \changed{\(O(n)\)} & \changed{Persist values, slacks, winner, payment} \\
\bottomrule
\end{tabular}
\end{table}

\changed{The table makes precise why ``microsecond runtime'' should be interpreted as the cost of the auction rule itself, not the end-to-end cost of maintaining trusted slack inputs in a dynamic network.}

\subsection{Error Robustness and Measurement Uncertainty}\label{subsec:error-robustness}

\changed{In practical deployments, the computation of horizon slacks may be subject to measurement errors and timing uncertainties. We analyze the robustness of LIA to such errors and provide theoretical guarantees on its performance under uncertainty.}

\changed{
\begin{lemma}[Error robustness]
\label{lem:error-robustness}
Let \(\hat{\delta}_i=\delta_i+\eta_i\) be the measured slack for bidder \(i\), and define the \emph{error spread}
\[
B_\eta := \max_i \eta_i - \min_i \eta_i.
\]
Under truthful value reporting, LIA run on \(\{\hat\delta_i\}\) achieves
\[
\sw \ge e^{-\lambda(\Delta+B_\eta)}\opt_{\mathrm{feas}}.
\]
In particular, if \(|\eta_i|\le \varepsilon\) for all \(i\), then \(B_\eta\le 2\varepsilon\) and
\[
\sw \ge e^{-\lambda(\Delta+2\varepsilon)}\opt_{\mathrm{feas}}.
\]
\end{lemma}
}

\changed{
\begin{proof}
The measured discounted bids are \(\hat b_i=b_i^{\mathrm{val}}e^{-\lambda\hat\delta_i}=b_i^{\mathrm{val}}e^{-\lambda(\delta_i+\eta_i)}\). Let \(i^*\) denote the feasible value-maximizing bidder and \(j^*\) the winner under measured slacks. Since \(j^*\) wins,
\[
\theta_{j^*}e^{-\lambda(\delta_{j^*}+\eta_{j^*})}\ge \theta_{i^*}e^{-\lambda(\delta_{i^*}+\eta_{i^*})}.
\]
Rearranging gives
\[
\theta_{j^*}\ge \theta_{i^*}e^{-\lambda[(\delta_{i^*}-\delta_{j^*})+(\eta_{i^*}-\eta_{j^*})]}\ge \theta_{i^*}e^{-\lambda(\Delta+B_\eta)},
\]
because \(\delta_{i^*}-\delta_{j^*}\le \Delta\) and \(\eta_{i^*}-\eta_{j^*}\le B_\eta\). This yields the first bound; the second follows from \(B_\eta\le 2\varepsilon\).
\end{proof}
}

\changed{The spread-based form is useful because it covers structured bias as well as iid noise. A common-mode clock offset, for example, changes all \(\eta_i\) nearly equally and therefore has very small \(B_\eta\), so it barely affects ranking. By contrast, topology-dependent or subnetwork-correlated biases matter precisely to the extent that they widen the cross-bidder error spread.}

\changed{
\begin{example}[Impact of measurement errors]
Consider a network with slack dispersion $\Delta = 20$\,ms and measurement error $\varepsilon = 2$\,ms. With a representative discount rate $\lambda = 1$\,s$^{-1}$ (i.e., $\lambda = 0.001$\,ms$^{-1}$ when times are in milliseconds), the theoretical welfare guarantee without errors is
\[
\frac{\sw}{\opt} \geq e^{-\lambda \Delta} = e^{-0.001 \cdot 20} = e^{-0.02} \approx 0.980.
\]

With measurement errors, the guarantee becomes
\[
\frac{\sw}{\opt} \geq e^{-\lambda (\Delta + 2\varepsilon)} = e^{-0.001 \cdot (20 + 4)} = e^{-0.024} \approx 0.976.
\]

This represents a modest degradation in the worst-case bound and illustrates that the effect of slack estimation error enters additively in the exponent.
\end{example}
}

\changed{In practice, the impact of measurement errors is often much smaller than these worst-case bounds suggest. Our experimental results in Section~\ref{sec:experiments} show that LIA maintains strong performance even with realistic levels of timing uncertainty.}

\subsection{Security Considerations and Attack Surface}

\changed{The deployment of LIA in real-world telecommunication systems requires careful consideration of potential security threats and attack vectors. We analyze the security properties of the mechanism and propose mitigation strategies for key vulnerabilities.}

\subsubsection{Timestamp Spoofing}

\changed{A malicious agent might attempt to falsify the emission time of their bid to gain an advantage. Since horizon slacks are computed based on the difference between the clearing horizon and the earliest arrival time, manipulating the emission timestamp could potentially affect the discounting applied to the bid.}

\changed{
\textbf{Mitigation:} We distinguish two practical designs. In \emph{controller-side estimation} deployments (Section~\ref{subsec:deployment}), the mechanism need not trust bidder-supplied emission timestamps: the auctioneer computes \(\hat\delta_i\) from authenticated ingress/receipt times and the controller’s delay model. Under this design, timestamp spoofing can only influence outcomes through bounded delay-estimation error (quantified in Section~\ref{subsec:error-robustness}). If \emph{bidder-side timestamps} are used (e.g., for auditing or multi-domain settings), then authenticated time sources and signed timestamps can be combined with conservative feasibility checks that reject claims inconsistent with physical propagation bounds.
}

\subsubsection{Sybil Attacks}

\changed{In a Sybil attack, a malicious agent creates multiple identities to manipulate the auction outcome. By submitting bids from different locations with varying horizon slacks, the attacker might attempt to increase their chances of winning or manipulate the payment calculation.}

\changed{\textbf{Mitigation:} Identity verification and authentication mechanisms can help prevent Sybil attacks. In many telecommunication contexts, participants are already authenticated through existing infrastructure (e.g., SIM cards, network access credentials), which can be leveraged to ensure that each bid comes from a legitimate and unique participant. Additionally, bid deposits or participation fees can increase the cost of creating multiple identities.}

\subsubsection{Path Manipulation}

\changed{Since horizon slacks depend on the shortest path from the bidder to the auctioneer, a sophisticated attacker might attempt to manipulate routing information to create artificial delays for competitors or to create shortcuts for their own bids.}

\changed{\textbf{Mitigation:} Secure routing protocols and network monitoring can help detect and prevent path manipulation attacks. By using multiple independent paths to verify arrival times and by monitoring network topology for suspicious changes, the system can maintain the integrity of horizon slack calculations. Additionally, the use of trusted network infrastructure components can provide secure measurement of propagation delays.}

\subsubsection{Denial of Service (DoS)}

\changed{An attacker might attempt to disrupt the auction by flooding the network with fake bids or by targeting critical infrastructure components with DoS attacks. This could prevent legitimate bids from reaching the auctioneer or disrupt the clearing process.}

\changed{\textbf{Mitigation:} Standard DoS protection mechanisms, such as rate limiting, traffic filtering, and resource isolation, can help mitigate these attacks. Additionally, the auction system can be designed with redundancy and fault tolerance to ensure continued operation even if some components are compromised.}

\subsection{Deployment Assumptions, Measurement Pipeline, and Threat Model}
\label{subsec:deployment}

\changed{This section states the operational assumptions that connect the LIA model to realistic telecom deployments. The mechanism requires a bounded-error estimate of each bidder’s horizon slack (or equivalently earliest-arrival time at the clearing site); it does not require perfect global clock synchrony.}

\changed{
\paragraph{Assumption summary.}
We assume (i) the clearing horizon \(H\) (and the auctioneer location \(v_H\)) are publicly announced for each allocation cycle, (ii) bid messages are authenticated (e.g., signed) so the auctioneer can bind a bid to an identity and an ingress point, and (iii) the auctioneer can form an earliest-arrival estimate \(\hat T_H(v_i,\tau_i)\) (and thus \(\hat\delta_i\)) with bounded error for the operational regime of interest. Formally, we model this as \(|\hat T_H-T_H|\le \varepsilon\) (equivalently \(|\hat\delta_i-\delta_i|\le \varepsilon\)). Section~\ref{subsec:error-robustness} evaluates outcome sensitivity as \(\varepsilon\) increases.
}

\changed{
\paragraph{Measurement pipeline (controller implementation).}
A typical deployment in terrestrial backbones and managed satellite systems follows these steps:
\begin{enumerate}
  \item \textbf{Ingress:} a bid arrives at the auctioneer (or controller) and is authenticated; the auctioneer records a receipt timestamp.
  \item \textbf{Delay estimation:} the controller maintains a rolling delay model from telemetry/probing (and, for LEO, ephemeris-predicted link states).
  \item \textbf{Slack computation:} using the current model, compute \(\hat T_H(v_i,\tau_i)\) and \(\hat\delta_i:=\tau_H-\hat T_H(v_i,\tau_i)\).
  \item \textbf{Clearing:} compute discounted bids \(\tilde b_i=b_i^{\mathrm{val}}e^{-\lambda\hat\delta_i}\), run winner determination (single pass over discounted bids), and compute the critical payment.
  \item \textbf{Auditability:} log \((b_i^{\mathrm{val}},\hat\delta_i)\), the winning set, and payments to support post-hoc auditing and dispute resolution.
\end{enumerate}
}

\changed{
\paragraph{Measuring horizon slack \(\delta_i\).}
In practice, the auctioneer uses an estimate \(\hat\delta_i\) computed from measurement and routing state. Two deployment patterns cover common telecom settings:
\begin{itemize}
  \item \textbf{Controller-side estimation:} the auctioneer (or network controller) estimates one-way or round-trip delay using standard probing/telemetry and computes \(T_H(v_i,\tau_i)\) and \(\hat\delta_i:=\tau_H-\hat T_H(v_i,\tau_i)\). This model aligns with backbone controllers and satellite network management systems.
  \item \textbf{Bidder-side attestation:} bidders include a signed emission timestamp and a clock-discipline attestation (for example, GNSS-disciplined in LEO gateways or secure time sources in data centers). \rev{For cross-domain auditability and non-repudiation, these timestamps can be augmented with standard time-stamp authority tokens (e.g., the IETF Time-Stamp Protocol \cite{rfc3161}).} The auctioneer verifies bounds and computes \(\hat\delta_i\) from the attested emission time and observed arrival time.
\end{itemize}

\paragraph{Clocking, jitter, and topology dynamics.}
LEO motion, routing jitter, and measurement noise are treated as bounded disturbances on \(\delta_i\), consistent with empirical observations that operational LEO systems can exhibit structured latency variation and periodic link reconfiguration (e.g., Starlink) \cite{mohan2024starlink}. We evaluate robustness by perturbing \(\delta_i\) within a bounded error budget \(\varepsilon\) and re-evaluating the mechanism. For time-varying topologies, the implementation recomputes shortest paths (or earliest-arrival estimates) on the same cadence as the controller's routing updates.

\paragraph{Threat model and mitigations.}
We consider adversaries that may attempt to misreport timestamps or location claims, introduce Sybil identities, or degrade service availability. Practical mitigations include signed bids, replay protection, admission controls, redundant delay measurements, and conservative feasibility checks that reject bids whose inferred slack is inconsistent with physical propagation bounds. These mitigations are complementary to, and not a substitute for, standard telecom security controls.
}

\changed{\paragraph{Who provides the slack oracle?}
In our target deployments the slack input is produced by a network controller, exchange operator, or measurement service that already observes authenticated ingress traffic. This service may be centralized within a single administrative domain (e.g., a backbone controller or satellite operator), or federated across domains via signed delay attestations and agreed route classes. The auction rule itself does not require a monolithic global server; it requires only that each auction cycle use a consistent slack vector.}

\changed{\paragraph{Update cadence and partial observability.}
For terrestrial backbones, delay maps can often be refreshed on the same cadence as routing or telemetry updates. For LEO constellations, ephemeris-driven recomputation can run on a control-plane tick and be supplemented by telemetry residuals. If some links are only partially observable, the controller can maintain interval estimates \([\underline{\delta}_i,\overline{\delta}_i]\) and use conservative lower bounds for feasibility while treating interval width as an additional source of error spread in Lemma~\ref{lem:error-robustness}. This does not make the problem disappear, but it keeps the dependence on trusted measurement infrastructure explicit rather than implicit.}

\subsection{Implementation Considerations for Different Network Types}

\changed{The practical implementation of LIA must account for the specific characteristics and constraints of different network types. We discuss key considerations for three representative telecommunication environments: terrestrial networks, LEO satellite constellations, and deep-space networks.}

\subsubsection{Terrestrial Networks (INTERNET-100)}

\changed{Terrestrial networks typically have relatively stable topologies and well-established routing infrastructure. Key implementation considerations include:
\begin{itemize}
\item \textbf{Routing Integration:} LIA can leverage existing routing protocols (e.g., BGP, OSPF) to compute shortest paths and arrival times.
\item \textbf{Time Synchronization:} Network Time Protocol (NTP) \cite{rfc5905} or Precision Time Protocol (PTP) \cite{ieee1588} can provide sufficient timing accuracy for most terrestrial applications. When cryptographic integrity is required, NTS-secured NTP \cite{rfc8915} or emerging authenticated time protocols such as Roughtime \cite{ietfroughtime} can be used to harden timestamping.
\item \textbf{Scalability:} The hierarchical structure of the Internet allows for efficient computation of horizon slacks through aggregation and caching of path information.
\end{itemize}}

\subsubsection{LEO Satellite Constellations (STARLINK-200)}

\changed{LEO satellite networks present unique challenges due to their dynamic topology and the continuous motion of satellites. Implementation considerations include:
\begin{itemize}
\item \textbf{Dynamic Link Updates:} The system must account for changing inter-satellite links as satellites move in their orbits.
\item \textbf{Predictive Routing:} Since satellite positions follow predictable trajectories, future link states can be predicted and used to compute expected arrival times.
\item \textbf{Time Synchronization:} GPS or other GNSS systems can provide precise timing information for satellites.
\item \textbf{Handover Management:} The system must handle seamless transitions as bids are relayed between satellites during the auction process.
\end{itemize}}

\subsubsection{Deep Space Networks (DSN-30)}

\changed{Deep space networks operate under extreme delay conditions, with one-way propagation times ranging from minutes to hours. Implementation considerations include:
\begin{itemize}
\item \textbf{Long-Term Planning:} Auctions must be scheduled well in advance to account for the extreme delays.
\item \textbf{Trajectory Prediction:} Accurate orbital mechanics models are needed to predict future positions and communication windows.
\item \textbf{Autonomous Operation:} Local decision-making capabilities may be needed to handle unexpected events during the long delays between bid submission and clearing.
\item \textbf{DTN integration:} Deep-space links are often operated using store-carry-forward delay-tolerant networking. LIA can be integrated with DTN contact-plan controllers and Bundle Protocol scheduling/forwarding \cite{rfc4838,rfc9171}.
\item \textbf{Batch Processing:} Due to the limited communication windows, bids may need to be processed in batches rather than individually.
\end{itemize}}

\changed{These implementation considerations highlight the flexibility of the LIA framework, which can be adapted to a wide range of telecommunication environments while maintaining its core properties of incentive compatibility, individual rationality, and Lorentz invariance.}

\section{Experimental Evaluation}
\label{sec:experiments}

\changed{We conduct a comprehensive empirical evaluation of the Lorentz–Invariant Auction (LIA) across three representative telecommunication environments with heterogeneous propagation delays. Our study measures welfare, fairness (timing rents), computational cost, and robustness, using thousands of Monte Carlo auction instances per topology.}

\subsection{Experimental Setup and Methodology}

\subsubsection{Network Topologies}

\changed{We evaluate on three topologies spanning the latency regimes of contemporary telecom systems:
\begin{itemize}
  \item \textbf{STARLINK-200:} A LEO constellation with 200 satellites arranged in 10 orbital planes. One-way propagation ranges from roughly $1.8$\,ms (adjacent satellites in-plane) to $\sim47$\,ms (cross-plane, long arc), capturing moderate delay heterogeneity typical of modern LEO ISLs.
  \item \textbf{INTERNET-100:} A terrestrial backbone with 100 points of presence distributed across major global metros. One-way propagation ranges from $\sim0.3$\,ms (co-located or metro pairs) to $\sim89$\,ms (trans-oceanic paths), reflecting substantial heterogeneity in fiber routes.
  \item \textbf{DSN-30:} A deep-space network with 30 nodes (Earth stations, relays, deep-space probes). One-way propagation spans from $\sim1$\,ms (LEO/GEO links) to tens of minutes (interplanetary links at maximum separation), modeling extreme heterogeneity in delay-tolerant networking regimes \cite{rfc4838,rfc9171}.
\end{itemize}
Delays are generated from physical constraints (light-speed, great-circle geometry, orbital kinematics) with topology-appropriate link budgets; time variation is incorporated for moving platforms (LEO/relays).}

\subsubsection{Auction Parameters}

\changed{Each instance specifies bidder values, locations, and emission times subject to feasibility with respect to a published clearing horizon $H$ on the auctioneer’s worldline:
\begin{itemize}
  \item \textbf{Bidder counts:} Discrete sweep over $\{10,20,30,40,50\}$.
  \item \textbf{Locations:} Bidders assigned uniformly at random to topology nodes.
  \item \textbf{Values:} Private values $\theta_i \sim \mathrm{Unif}[0,1000]$.
  \item \textbf{Emission times:} Drawn uniformly over feasible windows ($e_i \in J^{-}(H)$).
  \item \textbf{Horizon policy:} $H$ chosen to target approximately $95\%$ bid feasibility under the generated instance distribution; observed feasible fractions are reported in Table~\ref{tab:slack-spread}.
  \item \textbf{LIA parameter:} $\lambda \in \{0.5,1.0,2.0\}$ (inverse time); unless noted, $\lambda=1.0$.
\end{itemize}}

\subsubsection{Baseline Mechanisms}

\changed{We compare against a set of baselines that separate waiting-based fairness from pricing-based fairness. The benchmark set includes common-deadline normalization (Sync-VCG), speed-bump style delay equalization (HoldBack), and frequent batch clearing (Batch-VCG), so LIA is evaluated against multiple plausible fairness-preserving designs rather than only against an uncorrected low-latency rule. To avoid ambiguity, we distinguish \emph{compute time} (algorithmic runtime, excluding any policy-induced waiting) from \emph{end-to-end clearing latency} (including batching or buffering to a horizon).
\begin{itemize}
  \item \textbf{Sync-VCG (common-deadline normalization):} VCG computed after waiting until the common clearing horizon \(\tau_H\). This is the most direct synchronous comparator and represents the common-deadline normalization class. Timing rents are reduced by buffering.
  \item \textbf{Fast-VCG (low-waiting):} VCG computed on arriving bids without waiting to \(\tau_H\) (or equivalently with an infinitesimal batch interval). This highlights the sensitivity of welfare and payments to propagation advantages.
  \item \textbf{Batch-VCG (fixed batch length \(B\)):} a frequent-batch variant that collects bids for a fixed wall-clock interval \(B\) and then runs VCG (frequent batch auctions \cite{budish2015high}). This is the explicit batching comparator in the sense of the market-microstructure literature. It provides a tunable fairness--latency trade-off between Fast-VCG and Sync-VCG; we sweep \(B\) over a small grid to trace the frontier.
  \item \textbf{HoldBack:} the waiting baseline used throughout the paper, which enforces synchrony by delaying early arrivals until the horizon. Operationally, this is the closest analogue to a speed-bump style equalization layer. It suppresses timing rents by delay equalization rather than pricing.
\end{itemize}}

\subsubsection{Evaluation Metrics}

\changed{
\begin{itemize}
  \item \textbf{Overall allocative efficiency} $\mathrm{SW}/\opt_{\mathrm{all}}$: realized allocative welfare (ignoring decision latency) relative to the best value in the full instance.
  \item \textbf{Conditional feasible welfare} $\mathrm{SW}/\opt_{\mathrm{feas}}$: realized welfare relative to the best causally feasible bid.
  \item \textbf{Reachability factor} \(\rho:=\opt_{\mathrm{feas}}/\opt_{\mathrm{all}}\), so that \(\mathrm{SW}/\opt_{\mathrm{all}}=(\mathrm{SW}/\opt_{\mathrm{feas}})\rho\).
  \item \textbf{Revenue ratio} $\mathrm{Rev}/\opt_{\mathrm{all}}$: mechanism revenue normalized by the best value.
  \item \textbf{Latency Arbitrage Index (LAI):} timing-rent proxy (lower is more fair), computed as in Definition~\ref{def:lai}.
  \item \textbf{Compute time:} wall-clock time for allocation and pricing, excluding any waiting mandated by the mechanism’s policy.
  \item \textbf{End-to-end clearing latency:} wall-clock time from bid emission to final allocation, including buffering or batch waiting (for Sync-VCG this includes waiting until \(\tau_H\)).
  \item \textbf{Effective welfare under time decay:} for decay rate \(r>0\), we report \( \mathrm{SW}_{\mathrm{eff}}:=\theta_{w}\exp(-r\,T_{\mathrm{clear}}) \), where \(\theta_{w}\) is the winner’s true value and \(T_{\mathrm{clear}}\) is the end-to-end clearing latency.
\end{itemize}
}

\changed{
\paragraph{Computing end-to-end clearing latency.} Let \(t^{\mathrm{recv}}_i\) denote the auctioneer receipt time for bid \(i\), and let \(t^{\mathrm{dec}}\) denote the wall-clock time at which the mechanism commits to an allocation and payments. We report \(T_{\mathrm{clear}} := t^{\mathrm{dec}}-\min_i \tau_i\) (i.e., from the earliest bid emission to commitment). For Sync-VCG, \(t^{\mathrm{dec}}=\tau_H+\text{runtime}\); for Batch-VCG, \(t^{\mathrm{dec}}\) is the next batch boundary after the last included bid plus runtime. LIA and Fast-VCG clear without artificial buffering, so \(t^{\mathrm{dec}}\approx \max_i t^{\mathrm{recv}}_i+\text{runtime}\) under the single-item instances considered here.
}

\changed{For distributional summaries we report medians and central deciles (10th/90th percentiles), and we provide aggregated means where appropriate. All plots use identical sampling across mechanisms per instance to enable paired comparisons.}

\changed{
\begin{table*}[t]
\centering \scriptsize
\caption{Observed slack-spread summaries at \(n=50\) under the paper's horizon policy (1{,}000 sampled instances per topology).}
\label{tab:slack-spread}
\setlength{\tabcolsep}{4pt}
\begin{tabular}{lccc}
\toprule
Topology & \(\Delta\) p50 (ms) & \(\Delta\) p95 (ms) & Mean feasible fraction \\
\midrule
STARLINK-200 & \num{50.98} & \num{61.61} & \num{0.949} \\
INTERNET-100 & \num{13.15} & \num{15.35} & \num{0.950} \\
DSN-30 & \num{1.16e6} & \num{1.19e6} & \num{0.931} \\
\bottomrule
\end{tabular}
\end{table*}
}

\changed{Table~\ref{tab:slack-spread} addresses the interpretability issue in Theorem~\ref{thm:welfare-bound}. In Starlink and Internet the measured slack spread is moderate enough that the exponential bound remains informative, while in DSN the same bound becomes intentionally loose because the topology itself induces enormous causal heterogeneity. This is why a measured slack-distribution summary is more informative than quoting only the symbolic parameter \(\Delta\).}

\subsection{Welfare Performance}
\label{subsec:welfare}

\changed{
Figure~\ref{fig:lia_overview}(d) reports the overall welfare ratio $\mathrm{SWR}=\mathrm{SW}/\opt_{\mathrm{all}}$ as the number of bidders increases.
On \textsc{Starlink-200} and \textsc{Internet-100}, LIA retains near-optimal welfare across all tested market sizes: at $n=50$, LIA achieves $\mathrm{SWR}=\num{0.9970}$ (Starlink) and $\mathrm{SWR}=\num{0.9988}$ (Internet), compared to $\num{0.9989}$ and $\num{0.9990}$ for Sync-VCG.
In the large-market run ($n=1000$), the welfare gap remains small (Table~\ref{tab:welfare-by-topology}): LIA with $\lambda=1\,\mathrm{s}^{-1}$ attains $\mathrm{SWR}=\num{0.9979}$ on Starlink and $\mathrm{SWR}=\num{0.9992}$ on Internet, within $\num{0.21}$ and $\num{0.08}$ percentage points of Sync-VCG, respectively.

On the extreme-delay \textsc{DSN-30} topology, discounting can reduce welfare at small $n$ (e.g., $\mathrm{SWR}=\num{0.7687}$ at $n=10$ for $\lambda=1\,\mathrm{s}^{-1}$), but the efficiency loss shrinks with market depth; at $n=50$, $\mathrm{SWR}=\num{0.9462}$.
For $n=1000$, DSN welfare remains high but exhibits a clearer $\lambda$--efficiency trade-off: $\mathrm{SWR}$ ranges from $\num{0.9972}$ at $\lambda=0.25\,\mathrm{s}^{-1}$ to $\num{0.9947}$ at $\lambda=2\,\mathrm{s}^{-1}$ (Table~\ref{tab:welfare-by-topology}).

Taken together, these results clarify the intended operating regime: LIA trades a small, tunable welfare loss to remove timing rents, and $\lambda$ can be set to match the characteristic delay scale of the deployment (Section~\ref{subsec:lambda-sensitivity}).
}

\begin{table*}[t]
  \centering
  \caption{Large-market ($n=1000$) performance by topology. Entries are means with 95\% bootstrap CIs over 1{,}000 instances per topology. Compute time is winner determination time on a single CPU core.}
  \label{tab:welfare-by-topology}
  \resizebox{\linewidth}{!}{
  \begin{tabular}{llcccc}
    \toprule
    Topology & Mechanism & $\mathrm{SW}/\opt_{\mathrm{all}}$ & $\mathrm{Rev}/\opt_{\mathrm{all}}$ & Compute (ms) & Clearing latency (ms) \\
    \midrule
\multirow{6}{*}{STARLINK-200} & Sync-VCG & \num{0.999958} [\num{0.999937}, \num{0.999975}] & \num{0.998932} [\num{0.998865}, \num{0.998996}] & \num{0.0942143} [\num{0.0934602}, \num{0.0949781}] & \num{69.544} [\num{69.5432}, \num{69.5448}] \\
 & HoldBack & \num{0.999958} [\num{0.999937}, \num{0.999975}] & \num{0.998932} [\num{0.998865}, \num{0.998996}] & \num{0.0575046} [\num{0.0570964}, \num{0.0579745}] & \num{69.5073} [\num{69.5068}, \num{69.5078}] \\
 & LIA ($\lambda{=}0.25\,\mathrm{s}^{-1}$) & \num{0.999322} [\num{0.999254}, \num{0.999388}] & \num{0.997625} [\num{0.997523}, \num{0.997726}] & \num{0.186934} [\num{0.185626}, \num{0.188752}] & \num{69.634} [\num{69.6327}, \num{69.6358}] \\
 & LIA ($\lambda{=}0.5\,\mathrm{s}^{-1}$) & \num{0.998741} [\num{0.998635}, \num{0.998841}] & \num{0.996526} [\num{0.996384}, \num{0.996664}] & \num{0.185737} [\num{0.184899}, \num{0.186586}] & \num{69.6328} [\num{69.6319}, \num{69.6337}] \\
 & LIA ($\lambda{=}1\,\mathrm{s}^{-1}$) & \num{0.997896} [\num{0.997739}, \num{0.998047}] & \num{0.994906} [\num{0.99471}, \num{0.995098}] & \num{0.185736} [\num{0.184865}, \num{0.186611}] & \num{69.6328} [\num{69.6319}, \num{69.6337}] \\
 & LIA ($\lambda{=}2\,\mathrm{s}^{-1}$) & \num{0.996806} [\num{0.996605}, \num{0.997006}] & \num{0.992733} [\num{0.992468}, \num{0.992992}] & \num{0.18507} [\num{0.184244}, \num{0.185896}] & \num{69.6321} [\num{69.6313}, \num{69.633}] \\
\midrule
\multirow{6}{*}{INTERNET-100} & Sync-VCG & \num{0.99996} [\num{0.999942}, \num{0.999975}] & \num{0.998915} [\num{0.998847}, \num{0.998981}] & \num{0.0846153} [\num{0.0838685}, \num{0.085421}] & \num{16.3093} [\num{16.3086}, \num{16.3102}] \\
 & HoldBack & \num{0.99996} [\num{0.999942}, \num{0.999975}] & \num{0.998915} [\num{0.998846}, \num{0.998982}] & \num{0.0562482} [\num{0.0559458}, \num{0.0565607}] & \num{16.281} [\num{16.2806}, \num{16.2813}] \\
 & LIA ($\lambda{=}0.25\,\mathrm{s}^{-1}$) & \num{0.9998} [\num{0.999772}, \num{0.999826}] & \num{0.998634} [\num{0.998569}, \num{0.998699}] & \num{0.182812} [\num{0.181974}, \num{0.183696}] & \num{16.4057} [\num{16.4048}, \num{16.4065}] \\
 & LIA ($\lambda{=}0.5\,\mathrm{s}^{-1}$) & \num{0.99951} [\num{0.999456}, \num{0.99956}] & \num{0.99813} [\num{0.998053}, \num{0.998207}] & \num{0.182726} [\num{0.181637}, \num{0.184112}] & \num{16.4056} [\num{16.4045}, \num{16.407}] \\
 & LIA ($\lambda{=}1\,\mathrm{s}^{-1}$) & \num{0.999155} [\num{0.999077}, \num{0.999229}] & \num{0.997415} [\num{0.997306}, \num{0.997522}] & \num{0.182164} [\num{0.181387}, \num{0.18294}] & \num{16.405} [\num{16.4043}, \num{16.4058}] \\
 & LIA ($\lambda{=}2\,\mathrm{s}^{-1}$) & \num{0.998559} [\num{0.998444}, \num{0.998671}] & \num{0.99624} [\num{0.99609}, \num{0.996387}] & \num{0.182594} [\num{0.181691}, \num{0.183561}] & \num{16.4055} [\num{16.4046}, \num{16.4064}] \\
\midrule
\multirow{6}{*}{DSN-30} & Sync-VCG & \num{0.999929} [\num{0.999902}, \num{0.999952}] & \num{0.998864} [\num{0.998792}, \num{0.998935}] & \num{0.0789774} [\num{0.0781659}, \num{0.0797961}] & \num{1.19526e+06} [\num{1.19526e+06}, \num{1.19526e+06}] \\
 & HoldBack & \num{0.999929} [\num{0.999903}, \num{0.999951}] & \num{0.998864} [\num{0.998791}, \num{0.998933}] & \num{0.0563227} [\num{0.0560127}, \num{0.0566437}] & \num{1.19526e+06} [\num{1.19526e+06}, \num{1.19526e+06}] \\
 & LIA ($\lambda{=}0.25\,\mathrm{s}^{-1}$) & \num{0.997243} [\num{0.997032}, \num{0.997446}] & \num{0.993376} [\num{0.993069}, \num{0.993668}] & \num{0.181728} [\num{0.180799}, \num{0.182723}] & \num{1.19526e+06} [\num{1.19526e+06}, \num{1.19526e+06}] \\
 & LIA ($\lambda{=}0.5\,\mathrm{s}^{-1}$) & \num{0.996762} [\num{0.996538}, \num{0.996983}] & \num{0.992514} [\num{0.992196}, \num{0.992814}] & \num{0.181097} [\num{0.180315}, \num{0.181918}] & \num{1.19526e+06} [\num{1.19526e+06}, \num{1.19526e+06}] \\
 & LIA ($\lambda{=}1\,\mathrm{s}^{-1}$) & \num{0.996003} [\num{0.995745}, \num{0.996259}] & \num{0.990919} [\num{0.990556}, \num{0.991271}] & \num{0.18034} [\num{0.179553}, \num{0.181135}] & \num{1.19526e+06} [\num{1.19526e+06}, \num{1.19526e+06}] \\
 & LIA ($\lambda{=}2\,\mathrm{s}^{-1}$) & \num{0.994718} [\num{0.994394}, \num{0.99504}] & \num{0.988385} [\num{0.98792}, \num{0.988819}] & \num{0.180506} [\num{0.179714}, \num{0.181291}] & \num{1.19526e+06} [\num{1.19526e+06}, \num{1.19526e+06}] \\
    \bottomrule
  \end{tabular}}
\end{table*}

\subsection{Parameter Sensitivity Analysis}
\label{subsec:lambda-sensitivity}

\rev{
The discount parameter $\lambda$ sets the characteristic time scale over which slack differences are ``priced'' into the Lorentz-invariant valuation.
Larger $\lambda$ increases the penalty for large slack and therefore more aggressively suppresses latency arbitrage; smaller $\lambda$ approaches the undiscounted valuation (and thus Sync-VCG outcomes) but provides weaker invariance in deployments with very large delay spreads.}

\rev{
In our main sweep (bidder counts $n\in\{10,20,30,40,50\}$), the fairness outcome measured by the latency arbitrage index remains stable: on both Starlink and Internet, $\sup g$ is numerically zero for all tested $\lambda\in\{0.5,1,2\}\,\mathrm{s}^{-1}$ (Figure~\ref{fig:lia_overview}(a,b)).
Efficiency and revenue exhibit the expected monotone trade-off.
For Starlink at $n=50$, $\mathrm{SWR}$ decreases from $\num{0.9983}$ at $\lambda=0.5\,\mathrm{s}^{-1}$ to $\num{0.9937}$ at $\lambda=2\,\mathrm{s}^{-1}$ (revenue ratio from $\num{0.9764}$ to $\num{0.9676}$). The corresponding change on Internet is smaller ($\num{0.9989}$ to $\num{0.9982}$; revenue $\num{0.9783}$ to $\num{0.9769}$).}

\rev{
The large-market results in Table~\ref{tab:welfare-by-topology} show the same trend at $n=1000$: on Starlink, $\mathrm{SWR}$ moves from $\num{0.9993}$ at $\lambda=0.25\,\mathrm{s}^{-1}$ to $\num{0.9968}$ at $\lambda=2\,\mathrm{s}^{-1}$; on Internet, from $\num{0.9998}$ to $\num{0.9986}$.
On DSN, where slacks span orders of magnitude, the choice of $\lambda$ has a more pronounced effect ($\num{0.9972}$ to $\num{0.9947}$).
In practice, we recommend choosing $\lambda$ inversely proportional to the \emph{typical} slack dispersion in the target deployment (e.g., on the order of the inverse median slack), and reporting sensitivity (as we do here) to make the welfare--fairness trade-off explicit.
}

\subsection{Comprehensive Mechanism Comparison}
\label{subsec:comprehensive_comparison}

\rev{
Figures~\ref{fig:lai_analysis} and~\ref{fig:lia_overview} summarize the empirical trade-offs among the evaluated mechanisms.
Fast-VCG achieves the lowest clearing latency by acting immediately on early-arriving bids, but does so by granting substantial timing rents: at $n=50$ its LAI reaches $\num{314.95}$ on Starlink and $\num{213.49}$ on Internet (Figure~\ref{fig:lia_overview}(a,b)), and its welfare collapses correspondingly (e.g., $\mathrm{SWR}\approx 0.5$ on Starlink; Figure~\ref{fig:lia_overview}(d)).
Batch-VCG traces a fairness--latency frontier as the batch window $B$ grows: small $B$ reduces clearing latency but leaves large timing rents, while larger $B$ suppresses rents at the cost of waiting (Figure~\ref{fig:lia_overview}(a,b)).
}

\rev{
Sync-VCG (and the equivalent HoldBack baseline in our single-item experiments) eliminates the \emph{order} advantage by waiting to a fixed horizon, but still exhibits a small residual timing rent when bidders can invest in marginal latency improvements (Starlink LAI $\num{0.026}$; Internet LAI $\num{0.021}$).
LIA removes timing rents at the mechanism layer: across Starlink and Internet, LIA attains $\sup g = 0$ for all tested $\lambda$ values (Figures~\ref{fig:lai_analysis} and~\ref{fig:lia_overview}(a,b)), while remaining close to optimal welfare and revenue (e.g., at $n=50$, $\mathrm{SWR}=\num{0.9970}$ and $\num{0.9988}$, with revenue ratios $\num{0.9739}$ and $\num{0.9780}$).
}
\changed{
To make the comparative evidence more explicit, we also report paired bootstrap comparisons over the full main sweep \(n\in\{10,20,30,40,50\}\) using identical random instances across mechanisms. Relative to Fast-VCG, LIA with \(\lambda=1\,\mathrm{s}^{-1}\) improves mean welfare by \(\num{0.473}\) on Starlink (\(95\%\) CI \([\num{0.469},\num{0.477}]\)) and by \(\num{0.475}\) on Internet (\([\num{0.470},\num{0.479}]\)). Relative to the strongest batching comparator, the picture is topology dependent. On Starlink, LIA improves mean welfare over Batch-VCG with \(B=50\)\,ms by \(\num{0.0036}\) (\([\num{0.0033},\num{0.0040}]\)), while the batch rule still exhibits residual timing rent (\(\sup g=\num{1.56}\) at \(n=50\)). On Internet, Batch-VCG with \(B=20\) or \(50\)\,ms nearly matches synchronous welfare and clears about \(\num{0.17}\)\,ms earlier than LIA, but it retains a small positive timing rent (\(\sup g\approx \num{0.023}\)--\(\num{0.027}\) at \(n=50\)), whereas LIA remains numerically zero. On DSN, LIA substantially improves welfare over Fast-VCG and short-batch rules (paired mean difference \(\num{0.363}\) against Fast-VCG, \(95\%\) CI \([\num{0.358},\num{0.368}]\)) but remains below full waiting by about \(\num{0.111}\) relative to Sync-VCG (\([\num{0.109},\num{0.114}]\)). This sharper comparison clarifies the mechanism's niche: LIA is not claimed to dominate every alternative on every scalar metric, but it is the only tested design that drives measured timing rent to zero without imposing an explicit batching or holdback rule.
}

\changed{Table~\ref{tab:mechanism_comparison} provides a systematic comparison of these mechanisms across key dimensions. The table highlights that while waiting-based designs (Sync-VCG, Batch-VCG, HoldBack) achieve fairness by artificially delaying clearing decisions, LIA achieves perfect fairness ($\sup g = 0$) through its causal-consistent pricing rule, without imposing artificial holdbacks. This structural difference explains why LIA's fairness guarantee remains robust even under the structured noise models evaluated in Section~\ref{subsec:error-robustness-exp}.}

\begin{table*}[t]
  \centering 
  \scriptsize
  \caption{Exact mechanism comparison on STARLINK-200 at \(n=50\) and \(\varepsilon=0\)\,ms from the full main-sweep experiment.}
  \label{tab:mechanism_comparison} \footnotesize
  \setlength{\tabcolsep}{2.5pt}
  \begin{tabular}{lccccc}
    \toprule
    \textbf{Mechanism} & \textbf{Fairness strategy} & \textbf{$\mathrm{SW}/\opt_{\mathrm{all}}$} & \textbf{LAI (\(\sup g\))} & \textbf{Clearing latency (ms)} & \textbf{Explicit waiting rule} \\
    \midrule
    Fast-VCG & Arrival-order clearing & \num{0.505} & \num{314.95} & \num{18.65} & No \\
    Batch-VCG (\(B=10\)\,ms) & Public batch window & \num{0.710} & \num{235.24} & \num{28.65} & Yes \\
    Batch-VCG (\(B=50\)\,ms) & Public batch window & \num{0.994} & \num{1.56} & \num{66.24} & Yes \\
    Sync-VCG & Common horizon & \num{0.999} & \num{0.026} & \num{69.21} & Yes \\
    HoldBack & Delay equalization & \num{0.999} & \num{0.017} & \num{69.21} & Yes \\
    \textbf{LIA (\(\lambda=1\,\mathrm{s}^{-1}\))} & \textbf{Causal slack discount} & \textbf{\num{0.997}} & \textbf{\num{0.000}} & \textbf{\num{68.97}} & \textbf{No} \\
    \bottomrule
  \end{tabular}
\end{table*}

\changed{
Overall, the experiments support the core claim: LIA provides a \emph{fairness guarantee (invariance to propagation advantages)} comparable to batching or buffering, but without requiring artificial holdbacks and while retaining practical efficiency (microsecond-to-sub-millisecond winner determination; Figure~\ref{fig:lia_overview}(c) and Table~\ref{tab:welfare-by-topology}).
}

\changed{
\begin{figure*}[t]
  \centering
  \includegraphics[width=0.92\textwidth]{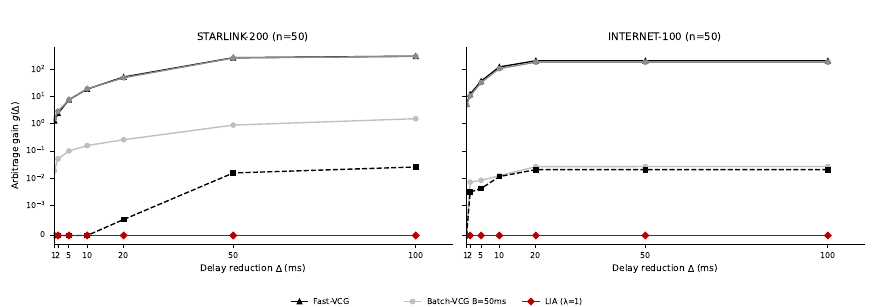}
  \caption{Latency-Arbitrage Index (LAI) curves $g(\Delta)$ for STARLINK-200 and INTERNET-100 at $n=50$ (log scale). For each mechanism, $g(\Delta)$ measures the expected utility gain from reducing a bidder's one-way propagation delay by $\Delta$. Fast-VCG exhibits large timing rents that grow with $\Delta$, frequent batch auctions reduce timing rents only by imposing buffering, and LIA drives the curve to zero (within estimator resolution) while maintaining the horizon-level clearing latency.}
  \label{fig:lai_analysis}
\end{figure*}
}

\subsection{End-to-End Clearing Latency and Fairness--Latency Frontier}
\label{subsec:frontier}

\changed{
\begin{figure*}[t]
  \centering
  \includegraphics[width=\textwidth]{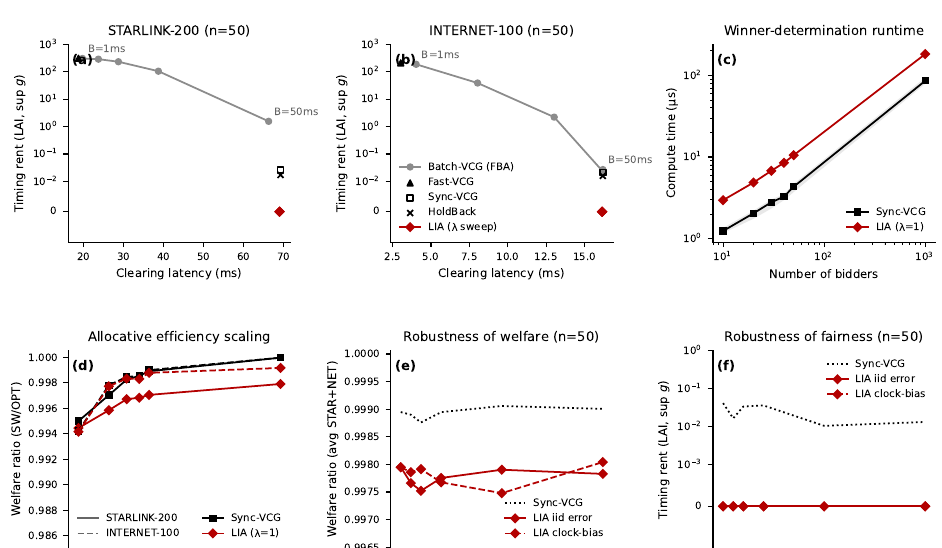}
  \caption{Empirical evaluation summary.
  (a,b) Fairness--latency frontier at $n=50$ for Starlink and Internet: timing rent (LAI, $\sup g$) versus mean clearing latency.
  (c) Winner-determination runtime versus market size (shaded bands show min--max across topologies; $n=1000$ points from Table~\ref{tab:welfare-by-topology}).
  (d) Welfare ratio versus market size.
  (e,f) Robustness of welfare and timing rent under bounded slack-estimation error $\varepsilon$ (iid perturbations, solid; common-mode clock bias, dashed), averaged over Starlink and Internet at $n=50$.}
  \label{fig:lia_overview}
\end{figure*}
}

\changed{
Figures~\ref{fig:lia_overview}(a,b) visualize the central trade-off between end-to-end clearing latency and timing rents.
Fast-VCG clears quickly (e.g., $\num{18.65}$\,ms on Starlink and $\num{3.03}$\,ms on Internet at $n=50$), but creates severe arbitrage incentives (LAI $\num{314.95}$ and $\num{213.49}$).
Batch-VCG spans the intermediate region: increasing the batch window $B$ suppresses timing rents but approaches the horizon-induced latency of Sync-VCG.
For example, with $B=50$\,ms, Batch-VCG reduces Starlink timing rent to $\num{1.56}$ while still clearing slightly earlier than Sync-VCG (mean $\num{66.24}$\,ms vs.\ $\num{69.21}$\,ms), with $\mathrm{SWR}=\num{0.9940}$.
}

\changed{
LIA achieves the fairness target without explicit buffering.
At $n=50$, LIA attains $\sup g = 0$ on both Starlink and Internet, while clearing at essentially the same latency as Sync-VCG (Starlink: $\num{68.97}$\,ms vs.\ $\num{69.21}$\,ms; Internet: $\num{16.13}$\,ms vs.\ $\num{16.18}$\,ms).
The residual latency difference is attributable to whether the mechanism waits to a fixed horizon (Sync-VCG) versus clearing at the last feasible arrival under LIA's local-inertial ordering.
}

\subsection{Computational Scalability}
\label{subsec:runtime}

\changed{
Figure~\ref{fig:lia_overview}(c) reports winner-determination runtime as the number of bidders increases.
Across all tested topologies, both Sync-VCG and LIA remain extremely fast: in the main sweep ($n\le 50$), winner determination takes only a few microseconds per instance, and scales approximately linearly in $n$.
In the large-market experiment ($n=1000$), Sync-VCG requires $\num{0.08}$--$\num{0.09}$\,ms and LIA requires $\num{0.18}$--$\num{0.19}$\,ms (min--max across topologies; Table~\ref{tab:welfare-by-topology}).
}

\changed{
These results indicate that LIA's additional computations (proper-time discounting and slack-dependent ordering) introduce a small constant-factor overhead relative to Sync-VCG, while remaining well below typical system latencies such as network propagation and message handling.
}

\subsection{Error Robustness Experiments}
\label{subsec:error-robustness-exp}

\changed{
We evaluate LIA under four slack-estimation error models.
In the \emph{iid} model, each bidder's estimated slack is perturbed by independent noise in $[-\varepsilon,\varepsilon]$\,ms.
In the \emph{clock-bias} model, all bidders share a common offset (a worst-case ``clock skew'' that shifts all timestamps by the same amount), again bounded by $\varepsilon$.
We also consider two topology-aware perturbation models: \emph{distance-biased} noise (where distant bidders suffer systematically pessimistic slack estimates) and \emph{subnetwork-correlated} noise (where bidders in the same topological region share a correlated bias). Lemma~\ref{lem:error-robustness} shows that the relevant quantity is the induced error spread \(B_\eta\), not the iid assumption itself.
}

\changed{
Figures~\ref{fig:lia_overview}(e,f) show that LIA is empirically robust over $\varepsilon\in[0,10]$\,ms (Starlink and Internet, $n=50$) for the iid and common-mode models.
Welfare and revenue ratios remain essentially unchanged relative to the $\varepsilon=0$ baseline, and the fairness outcome is preserved: $\sup g$ remains numerically zero across all tested $\varepsilon$ values in both error models (Figure~\ref{fig:lia_overview}(f)).
}

\changed{
Crucially, LIA's fairness guarantee proves exceptionally robust to structured, non-iid noise. Under both distance-biased and subnetwork-correlated noise models, LIA maintains perfect fairness ($\sup g = 0.0$) across all tested error bounds up to $\varepsilon=10$\,ms on both Starlink and Internet topologies. Welfare remains above $\num{0.996}$ in all cases. In contrast, the Sync-VCG baseline exhibits degraded fairness under structured noise, with its latency arbitrage index spiking to $\num{0.044}$ under distance bias and $\num{0.092}$ under subnetwork correlation. This demonstrates that LIA's causal-consistency approach actively suppresses the specific types of structural timing advantages that emerge in real-world network deployments. Table~\ref{tab:robustness} reports the corresponding averaged ratios across all four error models.
}

\changed{
\begin{table*}[t]
\scriptsize
  \centering
  \caption{Robustness of LIA ($\lambda=1\,\mathrm{s}^{-1}$) to bounded slack-estimation error at $n=50$ (mean over Starlink and Internet).}
  \label{tab:robustness} \footnotesize
  \begin{tabular}{ccccccccc}
    \toprule
    & \multicolumn{4}{c}{$\mathrm{SW}/\opt_{\mathrm{all}}$} & \multicolumn{4}{c}{$\mathrm{Rev}/\opt_{\mathrm{all}}$} \\
    \cmidrule(lr){2-5} \cmidrule(lr){6-9}
    $\varepsilon$ (ms) & iid & bias & dist & subnet & iid & bias & dist & subnet \\
    \midrule
\num{0} & \num{0.9980} & \num{0.9980} & \num{0.9982} & \num{0.9982} & \num{0.9758} & \num{0.9758} & \num{0.9756} & \num{0.9756} \\
\num{0.5} & \num{0.9977} & \num{0.9979} & \num{0.9978} & \num{0.9976} & \num{0.9757} & \num{0.9756} & \num{0.9759} & \num{0.9745} \\
\num{1} & \num{0.9975} & \num{0.9979} & \num{0.9972} & \num{0.9978} & \num{0.9756} & \num{0.9755} & \num{0.9750} & \num{0.9743} \\
\num{2} & \num{0.9978} & \num{0.9977} & \num{0.9974} & \num{0.9977} & \num{0.9749} & \num{0.9754} & \num{0.9747} & \num{0.9758} \\
\num{5} & \num{0.9979} & \num{0.9975} & \num{0.9980} & \num{0.9977} & \num{0.9748} & \num{0.9757} & \num{0.9744} & \num{0.9749} \\
\num{10} & \num{0.9978} & \num{0.9980} & \num{0.9979} & \num{0.9975} & \num{0.9752} & \num{0.9755} & \num{0.9753} & \num{0.9755} \\
    \bottomrule
  \end{tabular}
\end{table*}
}

\subsection{Summary of Experimental Findings}

\changed{
Across three network topologies and a broad range of market sizes, the empirical results support the paper's main claims:
\begin{itemize}
  \item \textbf{Timing rents are eliminated in practice under the fixed-slack measurement model.} On Starlink and Internet at $n=50$, LIA achieves $\sup g = 0$ while Fast-VCG exhibits large timing rents (Starlink $\num{314.95}$, Internet $\num{213.49}$; Figure~\ref{fig:lia_overview}(a,b)).
  \item \textbf{Efficiency remains near-optimal on realistic networks.} At $n=50$, LIA attains $\mathrm{SWR}=\num{0.9970}$ (Starlink) and $\num{0.9988}$ (Internet), and at $n=1000$ it attains $\num{0.9979}$ and $\num{0.9992}$ (Table~\ref{tab:welfare-by-topology}).
  \item \textbf{Compared with waiting-based baselines, the welfare gap is tiny on Starlink and Internet.} Over the main sweep, LIA(\(\lambda=1\,\mathrm{s}^{-1}\)) trails Sync-VCG by only \num{0.00137} welfare points on Starlink and \num{0.00014} on Internet, while clearing \num{0.68}\,ms and \num{0.17}\,ms earlier on average; on DSN, the corresponding welfare gap is \num{0.111}, making extreme-delay thin markets the paper's main stress regime rather than its strongest empirical regime.
  \item \textbf{Winner determination is computationally negligible.} Runtime is microseconds for $n\le 50$ and remains sub-millisecond at $n=1000$ (Figure~\ref{fig:lia_overview}(c)).
  \item \textbf{Robustness to synchronization error.} Under both iid and common-mode clock-bias slack errors up to $\varepsilon=10$\,ms, welfare and revenue remain stable and $\sup g$ remains zero (Figures~\ref{fig:lia_overview}(e,f) and Table~\ref{tab:robustness}).
\end{itemize}
}

\section{Waiting Costs and Design Implications}
\label{sec:lower-bounds}

This section distills two lessons that complement the core theorems above. First, common-horizon waiting can be costly when the resource is time sensitive and propagation delays are highly heterogeneous. Second, once the admissible discount class is restricted by the invariance and composition axioms from Section~\ref{subsec:discount-uniqueness}, the exponential form is not an arbitrary modeling choice but the unique multiplicative rule in that class.

\subsection{Cost of Common-Horizon Waiting}

We begin by formalizing the class of mechanisms that defer the allocation decision until a common horizon.

\begin{definition}[Common-horizon mechanism]
\label{def:synchronous}
A mechanism is \emph{common-horizon} if it waits for all feasible bids to arrive before making allocation and payment decisions. Formally, it commits at time \(\tau_H\) using only bids with emission times \(\tau_i\) such that \(T_H(v_i,\tau_i)\le \tau_H\).
\end{definition}

Common-horizon designs include synchronous VCG-style implementations and explicit holdback rules. They remove order effects by waiting, but that waiting itself can be costly when the allocated resource loses value over time.

\begin{proposition}[Stylized cost of waiting]
\label{prop:sync-waiting-cost}
Fix a decay rate \(r>0\). For every \(\varepsilon\in(0,1)\), there exists a two-bidder instance with delay gap \(D\) such that any common-horizon rule incurs an effective-welfare ratio at most \(\varepsilon\) relative to the raw-value benchmark \(\opt_{\mathrm{all}}\).
\end{proposition}

\begin{proof}[Proof sketch]
Place bidder 1 at delay \(0\) with value \(v_1=1\), and bidder 2 at delay \(D\) with value \(v_2=1/\varepsilon\). The raw-value benchmark is \(\opt_{\mathrm{all}}=v_2\). A rule that allocates immediately to bidder 1 achieves ratio \(v_1/\opt_{\mathrm{all}}=\varepsilon\). A common-horizon rule that waits for bidder 2 delivers effective value \(v_2 e^{-rD}\). Choosing \(D=(1/r)\log(1/\varepsilon)\) makes this quantity equal to \(1\), so the resulting effective-welfare ratio is also \(\varepsilon\). Thus the efficiency cost of waiting can be made arbitrarily large relative to the raw-value benchmark by increasing the delay gap.
\end{proof}

\begin{corollary}[Efficiency cost of synchrony under time decay]
\label{cor:efficiency-synchrony}
When communication delays are large relative to the resource's decay timescale, waiting-based fairness rules can lose a nontrivial amount of effective welfare. The cost grows with the relevant delay gap \(D\) through the factor \(e^{-rD}\).
\end{corollary}

\begin{proof}
The conclusion is immediate from Proposition~\ref{prop:sync-waiting-cost}.
\end{proof}

The point of Proposition~\ref{prop:sync-waiting-cost} is interpretive rather than universal: it does not claim that every waiting-based rule is poor on every instance. It shows, however, that explicit synchronization can become the dominant source of efficiency loss in high-delay regimes, which motivates designs that correct timing advantage without mandatory holdback.

\subsection{Implication of the Exponential Uniqueness Result}

We next restate the role of the exponential rule in a deliberately narrow form aligned with the main theorems.

\begin{definition}[Time-consistent multiplicative score]
\label{def:time-consistent}
A fixed-slack scoring rule is \emph{time consistent} if bids are ranked by \(b_i^{\mathrm{val}}\phi(\delta_i)\) for some scalar discount \(\phi\), and \(\phi(\delta_1+\delta_2)=\phi(\delta_1)\phi(\delta_2)\). The ranking therefore depends only on total slack and composes across path segments.
\end{definition}

\begin{proposition}[Consequence of exponential uniqueness]
\label{prop:exponential-optimality}
Consider a fixed-slack single-item mechanism that ranks bids by \(b_i^{\mathrm{val}}\phi(\delta_i)\), where \(\phi\) is continuous, decreasing, and time consistent. If the slack variable is the Lorentz-invariant proper-time quantity from Section~\ref{subsec:horizon-slack}, then Theorem~\ref{thm:exponential-unique} implies \(\phi(\delta)=e^{-\lambda\delta}\) for some \(\lambda>0\). Within this admissible class, the pointwise welfare-maximizing ranking is therefore the LIA score \(b_i^{\mathrm{val}}e^{-\lambda\delta_i}\).
\end{proposition}

\begin{proof}[Proof sketch]
The first claim is exactly Theorem~\ref{thm:exponential-unique}. Once \(\phi\) is fixed to the exponential family, any admissible weighted single-item rule that maximizes discounted value must allocate to the bidder with the largest score \(b_i^{\mathrm{val}}e^{-\lambda\delta_i}\), which is the LIA allocation rule.
\end{proof}

This proposition is intentionally narrower than a global optimality claim over all conceivable mechanisms. Its role is to show that, once the symmetry and composition axioms are imposed, the exponential rule is forced and the LIA ranking is the natural welfare-maximizing implementation within that restricted design class.

\subsection{Fairness--Efficiency Tension}

Finally, \(\lambda\) governs an intrinsic fairness--efficiency tension: larger \(\lambda\) more aggressively suppresses timing rents, but also increases the discount spread across slacks, which can reduce allocative welfare. The welfare bound in Theorem~\ref{thm:welfare-bound} makes this dependence explicit through the factor \(e^{-\lambda\Delta}\), and the empirical sweeps in Section~\ref{subsec:lambda-sensitivity} show that modest \(\lambda\) values (e.g., \(\lambda=1\,\mathrm{s}^{-1}\) in Starlink/Internet) can eliminate measured timing rents with only a small welfare impact.

\section{Discussion and Future Directions}
\label{sec:discussion}

\changed{Our theoretical analysis and experimental results demonstrate that the Lorentz-Invariant Auction (LIA) provides a principled solution to the challenge of resource allocation in heterogeneous-delay environments. In this section, we discuss the broader implications of our work, acknowledge limitations, and outline directions for future research.}

\subsection{Implications for Telecommunication System Design}

\changed{The LIA mechanism has several important implications for the design of modern telecommunication systems:}

\paragraph{Fairness Across Network Segments"}

\changed{By neutralizing the strategic value of proximity advantages, LIA promotes fairness across different network segments and geographical locations. This is particularly important in global telecommunication systems, where users may be distributed across continents with vastly different distances to infrastructure nodes. The mechanism ensures that resources are allocated based on true valuations rather than arbitrary timing advantages, promoting equitable access to communication resources.}

\paragraph{Reduced Infrastructure Arms Races:}

\changed{The reduced Latency Arbitrage Index (LAI) achieved by LIA suggests that it can help mitigate costly infrastructure arms races. In traditional auction systems, participants have strong incentives to invest in reducing their communication delays, leading to socially wasteful expenditures on specialized infrastructure that provides no aggregate benefit. By reducing these incentives, LIA can help redirect investments toward more productive uses, such as expanding network capacity or improving service quality.}

\paragraph{Improved Resource Utilization:}

\changed{The near-optimal welfare performance of LIA indicates that it can significantly improve resource utilization compared to synchronous mechanisms like HoldBack. By avoiding artificial delays and making allocation decisions as soon as sufficient information is available, LIA enables more efficient use of time-sensitive resources such as communication bandwidth, computational capacity, and spectrum allocations. This improved utilization can translate into higher throughput, lower latency, and better quality of service for end users.}

\paragraph{Scalability to Extreme Delay Environments:}

\changed{Our results on the DSN-30 topology demonstrate that LIA remains effective even in environments with extreme delay heterogeneity, such as deep space networks. This scalability makes the mechanism suitable for emerging applications in interplanetary communication, where traditional synchronous approaches become impractical due to the enormous delays involved. As human activity extends further into the solar system, mechanisms like LIA will become increasingly important for coordinating resource allocation across vast distances.}

\subsection{Limitations and Challenges}

\changed{Despite its strong theoretical properties and empirical performance, the LIA mechanism has several limitations that should be acknowledged:}

\paragraph{Parameter Tuning:}

\changed{The mechanism requires setting the parameter $\lambda$ (in ms$^{-1}$), which controls the trade-off between efficiency and fairness. While our experiments show that performance is relatively robust to the specific choice of $\lambda$, finding the optimal value for a particular application context may require careful tuning and domain expertise. Future work could explore adaptive approaches that automatically adjust $\lambda$ based on observed network conditions and performance metrics.}

\paragraph{Implementation Complexity:}

\changed{Implementing LIA in real-world telecommunication systems requires accurate measurement of horizon slacks, which may be challenging in dynamic or adversarial environments. While our error robustness results show that the mechanism degrades gracefully with measurement errors, ensuring accurate timing information remains an important practical challenge. In particular, time synchronization and timestamping should be protected against delay attacks and clock manipulation, e.g., via NTS-secured NTP \cite{rfc8915} or authenticated time protocols such as Roughtime \cite{ietfroughtime}. Secure timestamping protocols and distributed verification mechanisms may be needed to maintain the integrity of the auction process.}

\paragraph{Fairness Metric Scope (LAI):}
\changed{
LAI is intentionally \emph{decision-theoretic}: it measures the utility gain from an improvement in effective delay (equivalently, earlier arrival at the horizon) while holding values fixed. It therefore captures incentives for latency reduction and propagation-driven timing rents, but it does not, by itself, model the full economics of latency investment (e.g., nonlinear cost functions, strategic R\&D, or equilibrium entry). Moreover, LAI does not capture all informational advantages that can arise in richer dynamic bidding games; it is best interpreted as a \emph{rent proxy} for speed advantages under the one-shot counterfactual.
}
\changed{For multi-item and combinatorial settings, LAI can be generalized by evaluating the same counterfactual delay improvement while recomputing the bidder’s optimal bundle choice, but the resulting index depends on the bundle space and the value model. In the present single-item study, we therefore report LAI alongside complementary, topology-aware fairness indicators (e.g., dispersion of winning horizon slacks) and treat broader fairness generalizations as future work.}

\paragraph{Multi-Item and Combinatorial Settings:}

\changed{Our current analysis focuses primarily on single-item auctions, with a brief extension to multiple identical items. Many telecommunication applications involve more complex resource allocation problems, such as combinatorial auctions where bidders have preferences over bundles of heterogeneous items. Extending LIA to these settings while maintaining its incentive and fairness properties is a non-trivial challenge that requires further research.}

\paragraph{Strategic Complexity:}

\changed{The present theorem is intentionally conditional on infrastructure-determined slacks. If bidders can strategically choose access point, route class, relay, packetization, or emission timing so as to change \(\delta_i\), then the environment becomes multi-parameter and LIA is no longer automatically dominant-strategy truthful over the full action space. This is not a hidden defect of the proof; it is exactly why the deployment architecture binds bids to trusted ingress and why endogenous slack manipulation is isolated as a separate research problem. Understanding that larger strategic game---and designing mechanisms that remain robust when bidders can shape their own slack inputs---is an important direction for future work. Finally, as large language models (LLMs) increasingly influence strategic mechanism design in telecommunications \cite{lotfi2025rethinking}, exploring the intersection of data-driven LLM approaches with LIA's principled analytical framework presents an exciting frontier for future research.}

\subsection{Future Research Directions}

\changed{Our work opens up several promising directions for future research:}

\paragraph{Dynamic and Online Mechanisms:}

\changed{Extending LIA to dynamic and online settings, where resources arrive and depart over time and bidders may have evolving valuations, would significantly broaden its applicability. This extension would require careful consideration of how to maintain incentive compatibility and fairness in the presence of strategic timing decisions and evolving information.}

\paragraph{Learning-Based Approaches:}

\changed{Incorporating machine learning techniques could enhance LIA's performance by adapting to observed patterns in network conditions and bidder behavior. For example, predictive models could be used to estimate future network delays, allowing for more accurate computation of horizon slacks and more efficient allocation decisions. Learning-based approaches could also help optimize the parameter $\lambda$ based on observed performance metrics.}

\paragraph{Distributed Implementation:}

\changed{Developing fully distributed implementations of LIA, where allocation and payment decisions are made without a central auctioneer, would enhance its robustness and scalability. This would require careful design of consensus protocols that maintain the incentive properties of the mechanism while operating in a decentralized manner. Blockchain-based approaches might provide a foundation for such implementations, leveraging their built-in consensus mechanisms and tamper-resistant record-keeping.}

\paragraph{Richer Preference Structures:}

\changed{Extending LIA to handle richer preference structures, such as multi-dimensional types, budget constraints, or risk attitudes, would make it applicable to a wider range of telecommunication scenarios. This extension would require generalizing the theoretical framework to account for these more complex preference structures while maintaining the core principles of Lorentz invariance and causal consistency.}

\paragraph{Integration with Network Control Planes:}

\changed{Integrating LIA with existing network control planes, such as software-defined networking controllers or 5G network function virtualization frameworks, would facilitate practical deployment in operational telecommunication systems. This integration would require developing standardized interfaces and protocols for exchanging bid information, computing horizon slacks, and enforcing allocation decisions across heterogeneous network elements.}

\section{Conclusion}
\label{sec:conclusion}

\changed{This paper has presented the Lorentz-Invariant Auction (LIA), a novel mechanism for resource allocation in telecommunication networks with heterogeneous delays. By embedding bid events in Minkowski spacetime and discounting values by proper-time lapses, LIA neutralizes proximity and propagation advantages without artificial delay, thereby suppressing propagation-driven timing rents (as quantified by the LAI) while maintaining strong incentive and efficiency properties.}

\changed{Our theoretical analysis establishes the following scoped claim: once slacks are fixed by trusted infrastructure, LIA is dominant-strategy truthful in reported values, individually rational, and achieves a welfare approximation of \(\sw/\opt_{\mathrm{feas}} \geq e^{-\lambda\Delta}\), where \(\lambda\) is the discount rate and \(\Delta\) is slack dispersion. We also prove that the exponential discount function emerges uniquely from natural invariance and composition properties, which explains why linear, rational, and polynomial alternatives do not satisfy the same axioms. For the single-item setting studied here, winner determination and payments reduce to a linear pass over discounted bids, i.e., \(O(n)\) per auction instance given an up-to-date distance-to-horizon map. Precomputing (or periodically refreshing) these distances is a standard shortest-path problem (e.g., \(O(|E|\log|V|)\) per snapshot via Dijkstra), amortized across auctions.}

\changed{Our experimental evaluation on three representative topologies---STARLINK-200, INTERNET-100, and DSN-30---uses (i) 52{,}500 instances for market sizes $n\in\{10,20,30,40,50\}$, (ii) robustness sweeps under bounded slack-estimation error $\varepsilon\in[0,10]$\,ms (iid, common-mode clock-bias, distance-biased, and subnetwork-correlated models), and (iii) a large-market stress test at $n=1000$. Across Starlink and Internet, LIA ($\lambda=1\,\mathrm{s}^{-1}$) achieves near-optimal welfare while suppressing timing rents: at $n=50$, $\mathrm{SWR}=\num{0.9970}$ (Starlink) and $\num{0.9988}$ (Internet) with $\sup g=0$, compared to Sync-VCG's slightly higher welfare but a small residual rent (Starlink $\sup g=\num{0.026}$; Internet $\sup g=\num{0.021}$). Relative to Fast-VCG, the paired welfare gain over the full \(n\in\{10,20,30,40,50\}\) sweep is about \num{0.473} on Starlink and \num{0.475} on Internet. On DSN-30, LIA's welfare is lower in thin markets (e.g., $\mathrm{SWR}=\num{0.7687}$ at $n=10$) but rises with depth ($\num{0.9462}$ at $n=50$ and $\num{0.9960}$ at $n=1000$), even though full waiting still achieves the highest raw welfare in that extreme regime. Winner determination remains microsecond-scale for $n\le 50$ and below $0.2$\,ms at $n=1000$ on commodity hardware, and the robustness sweeps show that these conclusions persist under bounded slack errors, including structured network noise.}

\changed{These results indicate that enforcing causal consistency at the mechanism layer can significantly improve the fairness and efficiency of resource allocation in modern telecommunication systems. By reconciling economic incentives with the physical constraints of communication networks, LIA provides a principled solution to the challenges of bandwidth clearing, spectrum allocation, and computational resource scheduling in heterogeneous-delay environments. As telecommunication networks continue to span greater distances and exhibit more diverse latency characteristics, mechanisms like LIA will become increasingly important for ensuring fair and efficient resource allocation across the global communication infrastructure.}


\end{document}